\documentclass[10pt,twocolumn,twoside]{IEEEtran} 

\usepackage{amsmath,mathrsfs}
\usepackage{amsthm,tikz}
\usepackage{amssymb}
\usepackage{amsfonts}
\usepackage{amsxtra}     
\usepackage{epsfig}
\usepackage{verbatim}
\usepackage[ruled]{algorithm2e}
\usepackage{multirow}
\usepackage{caption,color}
\usepackage{url}

\DeclareMathOperator*{\sinc}{sinc}
\DeclareMathOperator*{\arcsine}{arcsin}
\DeclareMathOperator*{\psinc}{psinc}

\theoremstyle{plain}
\newtheorem{theorem}{Theorem} 

\makeatletter
\def\blfootnote{\xdef\@thefnmark{}\@footnotetext}
\makeatother

\theoremstyle{definition}
\newtheorem{definition}[theorem]{Definition}

\def\complex { \mathbb{C} }

\newcommand{\E}{\mathbf E}

\begin{document}

\title{Modal Analysis with Compressive Measurements}

\author{Jae Young~Park,
        Michael B.~Wakin,
        and~Anna C.~Gilbert
\thanks{JYP and ACG are with the Department of Electrical Engineering and Computer Science at the University of Michigan. Email: jaeypark,annacg@umich.edu.
MBW is with the Department of Electrical Engineering and Computer Science at the Colorado School of Mines. Email: mwakin@mines.edu. This work was partially supported by NSF grants CCF-1161233 and CIF-0910765, 
AFOSR grant FA9550-09-1-0465, 
and NSF CAREER grant CCF-1149225. 
}
}

\maketitle

\begin{abstract}
Structural Health Monitoring (SHM) systems are critical for monitoring aging infrastructure (such as buildings or bridges) in a cost-effective manner.
Such systems typically involve collections of battery-operated wireless sensors that sample vibration data over time.
After the data is transmitted to a central node, modal analysis can be used to detect damage in the structure.
In this paper, we propose and study three frameworks for Compressive Sensing (CS) in SHM systems; these methods are intended to minimize power consumption by allowing the data to be sampled and/or transmitted more efficiently.
At the central node, all of these frameworks involve a very simple technique for estimating the structure's mode shapes without requiring a traditional CS reconstruction of the vibration signals; all that is needed is to compute a simple Singular Value Decomposition.
We provide theoretical justification (including measurement bounds) for each of these techniques based on the equations of motion describing a simplified Multiple-Degree-Of-Freedom (MDOF) system, and we support our proposed techniques using simulations based on synthetic and real data.
%
%
%
\end{abstract}


\section{Introduction}

\subsection{Structural Health Monitoring systems}

Over the past decade, more than 5 million commercial buildings~\cite{cbecs01}, 130 million housing units~\cite{cbecs02}, and 0.6 million bridges~\cite{cbecs03} have been built in the United States.
In any structure, damage caused over time by continuous use is inevitable.
In order to maintain safely operable structures for as long as possible, periodic inspections are a must.
When damage is detected, some structures can be repaired, while others must be taken out of service immediately.

Due to the quantity, size, and complexity of structures, the task of inspection is labor intensive, costly, and time consuming.
Consequently, there have been significant efforts in the structural engineering community to automate this process.
Structural Health Monitoring (SHM) systems are precisely designed to address this issue.

Although the details of SHM systems vary, some features are shared among many methods~\cite{Sohn_Farrar_Hemez_Czarnecki_2001,Doebling_Farrar_Prime_Shevitz_1996}.
A typical SHM system monitors an in-service structure in real-time.
To do so, it makes use of a network of sensors installed on the structure to collect vibration data for damage detection.
This may include strain data, acceleration data, velocity data, or displacement data.
The acquired data from each sensor is transmitted over a network to the central data repository where damage detection algorithms are run to detect, localize, or classify possible damage in the structure. 

An important part of damage detection is a process called modal analysis.
This process is used to infer properties such as the modal frequencies, mode shapes, and modal damping ratios of the structure.
Such modal parameters describe the vibrational characteristics when external forces such as wind, earth quakes, or vehicle loadings are applied to the structure.
For example, if a structure is forced to vibrate close to a modal frequency, the shape of the structure's vibration will be dominated by the corresponding mode shape. This vibration will eventually die out in the absence of external force, and the modal damping ratio will determine rate of decay.

Many damage detection algorithms make use of modal parameters to detect, localize, and assess the severity of damage.
Briefly speaking, these methods rely on the notion that when a structure is damaged, its modal parameters will change.
Assuming that one has the modal parameters from the time when the structure was healthy, these can be compared to the current estimates of modal parameters to judge whether or not damage has occurred.
A comprehensive survey of damage detection methods is presented in~\cite{Sohn_Farrar_Hemez_Czarnecki_2001,Doebling_Farrar_Prime_Shevitz_1996}.

\subsection{Wireless SHM systems}

In the early designs of SHM systems, sensors were linked via coaxial cables that provided reliable communication to the central data repository.
Power cables were coupled with the coaxial cables to provide the power to run the sensors.
Unfortunately, it was often impractical to install a dense network of sensors along with all of the requisite cables.
When only a few sensors could be installed, the accuracy of damage detection and analysis would be sacrificed.

As a way to overcome this issue, researchers have proposed to deploy wireless sensors on structures to acquire and transmit data to the central repository.
With the rapid advancement of wireless technology and the ability to build sensors at low cost, it has become possible to quickly deploy a much denser network of sensors for a given budget.

The challenges of wireless sensor networks in SHM are quite different from those of wired sensor networks.
In particular, a primary challenge in a wireless network is to maximize the life expectancy of the batteries that power the sensors.
From this perspective, there are important questions that should be considered when designing a wireless sensor network for SHM:
Should we compress the time data before we transmit it to the central data repository?
This would allow us to transmit less information but would require additional power for processing.
Or, would we save more battery life by sending the raw data itself without processing it at all?

In \cite{Lynch15022007}, the authors presented thorough answers to these questions in real world applications involving a certain wireless sensor.
The authors found that on average, significant savings in battery can be expected by locally compressing and processing the data first compared to sending the raw time data.
The main reason for this was that the radio drains much more power (almost 3$\times$ more) than the on-board processor.
As long as the execution time of the local algorithm is reasonably short, it would be more beneficial to first process and compress, then transmit less information.

Another factor that contributes to the draining of the battery is the sampling process.
The Nyquist-Shannon sampling theorem states that when sampling a signal, the sampling rate should be proportional to the maximum frequency content in the signal.
For the purpose of modal analysis the highest frequency content is dictated by the highest modal frequency of the underlying structure.
Intuitively, structures that are stiff and light will generally posses high modal frequencies, and for such structures we may have to sample at relatively high
rates.
Unfortunately, faster sampling requires more power.
The total length of the sampling time span also plays an important role in optimizing the battery life.
Obviously, it would be best to sample only for a certain amount of time and turn the sensor off once a sufficient amount of data has been collected.

\subsection{Compressive Sensing in wireless SHM systems}

In light of these observations, we believe that wireless sensors equipped with Compressive Sensing (CS) architectures will be a perfect fit to improve the efficiency and accuracy of wireless SHM systems.
The theory of CS has in recent years offered the great promise of efficiently capturing essential signal information with low-rate sampling protocols, often below the minimum rate required by the Nyquist sampling theorem~\cite{Donoho06,Candes06,Wakin12}.
The use of CS can dramatically reduce both the complexity of a sensor and the amount of data that must be stored and/or transmitted downstream.

Employing CS in wireless SHM systems would allow compressed data to be acquired directly without the need for local processing.
This enables power savings in several ways: First, because the data can be collected at a rate lower than the Nyquist rate the physical demands on the sensing hardware will be lower.
Second, because there is no need for local processing, there is no computational burden at the sensor.
Third, because the data is acquired in compressed form, the amount of information that must be wirelessly transmitted to the central data repository is minimized.
Finally, all of the sensors can acquire the time data in a completely disjoint fashion,\footnote{In this paper, we assume that the time samples obtained at different sensors are all synchronized in time.} eliminating the need for sensors to communicate while compressing their data.

Mathematically, denoting the continuous-time displacement signal at each sensor as $v_l(t)$, where $l=1,\dots,N$, and $N$ denotes the number of nodes, a CS architecture employing non-uniform sampling (which is just one possible CS protocol~\cite{Wakin12}) would simply sample at time points $t_1,\dots,t_M$ and transmit the resulting samples $v_l(t_1),\dots,v_l(t_M)$.
The compression would come from the fact that $M$ is smaller than the number of Nyquist samples obtained over an appropriately chosen total sampling time span.\footnote{We use the word ``compression'' to reflect the fact that fewer samples are transmitted. A deeper analysis---beyond the scope of this paper---would consider specific binary encodings of the raw and CS data and quantify the actual compression savings in terms of bits.}
Typically, one would choose $t_1,\dots,t_M$ randomly.
%
%

Typically, when we undersample a signal compared to its Nyquist rate, we must solve an underdetermined system of equations in order to reconstruct the original signal (this requires a sparse model for the signal in some basis).
There have been several papers involving the application of CS to SHM systems.
In~\cite{CSDD2011}, the authors implement a prototype wireless sensor that computes the compressed measurements locally after the wireless sensor has acquired the time data.
By sending both compressed measurements and the raw time measurements to the central node, the authors evaluate the performance of CS reconstruction of the raw data.
It is shown that a number of measurements $M\approx 0.8L$ is needed for an accurate reconstruction of the raw data, where $L$ represents the length of the original Nyquist-rate sample vector.
Once the time data is reconstructed, conventional modal analysis techniques are carried out that aid in damage detection.
In \cite{YuequanBao01052011}, the authors also reconstruct the original time data from CS measurements obtained at each sensor.
Similar to~\cite{CSDD2011}, the authors claim that a relatively large number of measurements are required for accurate reconstruction.

The main reason that the above methods require so many measurements is because the signals are simply not sparse enough in the Discrete Wavelet Transform (DWT) or Discrete Fourier Transform (DFT).
This suggests that the model of sparsity may not be sufficiently strong for the individual signals that arise in modal analysis.
A joint signal model for the entire signal ensemble could serve as a better model to exploit the correlations that potentially exist across the signals.

At this point it is worth asking whether signal reconstruction is necessary at all when employing CS in wireless SHM systems.
The only reason that the above methods attempt to reconstruct the original signals is that conventional modal analysis tools are designed to operate on signal samples captured at or above the Nyquist rate.
There are two reasons to question whether this approach is optimal.
The first reason is the potential loss of performance.
The frequency content within the signals plays an important role in the subsequent modal analysis.
For example, a popular modal analysis algorithm proposed in~\cite{Brincker_Zhang_Andersen_2000} is very sensitive to the accuracy of each frequency component of the signal.
As we anticipate noise in our acquired signals, the reconstruction of each signal will also be noisy.
These errors will propagate to the modal analysis step, which could potentially lead to misidentification and errors in the estimated modal parameters.

The second reason is the added computational complexity.
Taking the above method for example, using an off-the-self CS reconstruction algorithm presented in~\cite{nt09}, the total computational complexity for $N$ signals will scale as $NL\log^2(L)$.
Added to this will be the computation required for the actual modal analysis.
One can view the reconstruction step as being somewhat wasteful as it is carried out only to enable the use of conventional modal analysis algorithms.

\subsection{Contributions}

The main objective of this paper is to propose a novel method for directly extracting the mode shapes from CS measurements without the need for reconstructing the individual time signals.
Our proposed method differentiates itself from the previously proposed CS-based methods in that it exploits the joint signal structure that arises from the equations describing a simplified (no damping and free-decay) Multiple-Degree-Of-Freedom (MDOF) system.

Our method is as simple as computing the Singular Value Decomposition (SVD) of the signal matrix obtained by stacking each sample vector $\{v_l\}=\{v_l(t_1),\dots,v_l(t_M)\}^*$ into a matrix denoted as $[V]=[\{v_1\},\dots,\{v_n\}]^*$. Here, the superscript $^*$ denotes the conjugate transpose operator (we explain the use of complex-valued data in Section~\ref{sec:analytic}).
We evaluate the performance of this method both when $t_1,\dots,t_M$ are chosen deterministically with uniform spacing and when $t_1,\dots,t_M$ are chosen uniformly at random.
We also evaluate the performance when we compute the SVD of $[Y]$ which is formed by choosing $t_1,\dots,t_M$ as uniform deterministic time points to form $[V]$ and subsequently multiplying $[V]$ by an $M\times M'$ ($M'<M$) random matrix $[\Phi]$ such that $[Y]=[V][\Phi]$.

For each sampling method, we give sufficient conditions on the required sampling rate, the total sampling time span, and the total number of measurements for accurate recovery of mode shape vectors.
Our analysis reveals that the requisite sampling rate for uniform sampling can be lower than the Nyquist rate, but that the required number of samples is structure dependent.
For random sampling, our analysis reveals that the required number of samples is structure independent and that we can achieve the same recovery guarantee as for uniform sampling, but the number of samples has a slightly increased dependence on the number of sensor nodes.
Finally, our analysis for the scheme involving uniform sampling followed by random matrix multiplication shows that the requisite number of measurements (the number of columns of $\Phi$) is dependent on the rank of $[V]$.
At the end of this paper, we present promising simulation results showing that our methods can accurately estimate the mode shapes using a number of samples or measurements that is only a small fraction of the original signal length.

\section{Background}

In this section, we give an introduction to the frequently used mathematical model that governs the motion of structures.
These equations form the basis of our proposed method.
We begin with the simple Single-Degree-Of-Freedom (SDOF) system and then move on to the MDOF system.
Following what is standard in the structural dynamics community, we use $\{x\}$ to denote a vector $x$ and $[A]$ to denote a matrix $A$.
We denote the $l$th entry of $\{x\}$ as $\{x\}(l)$, and the entry of $[A]$ in the $l$th row and $n$th column as $[A]_{l,n}$.
Furthermore, we reserve $i=\sqrt{-1}$ to denote the imaginary unit. 

\subsection{Single-degree-of-freedom system}

An SDOF system under no external force can be described by the following differential equation:
\begin{equation}
m\ddot{x}(t)+C\dot{x}(t)+k{x}(t)=0,
\label{eqn:sdof}
\end{equation}
where $m$, $C$, and $k$ denote the mass, damping, and stiffness parameters of the underlying system.
To solve for the displacement signal $x(t)$ that satisfies the above equation, let us assume a solution of the form $x(t)=Ae^{st}$, where $A,s\in\mathbb{C}$.
Then, $\ddot{x}(t)=As^2e^{st}$, $\dot{x}(t)=Ase^{st}$, $x(t)=Ae^{st}$, and by plugging in these expressions into equation~\eqref{eqn:sdof} we get $\left(ms^2+Cs+k\right)x(t)=0$.
Since this needs to be satisfied for all $t$, it must be that $ms^2+Cs+k=0$, and it is easy to see that $s=\frac{-C\pm\sqrt{C^2-4mk}}{2m}$.
In the structural dynamics community it is customary to rewrite this as
\begin{equation}
s=-\xi\omega_0\pm\omega_0\sqrt{\xi^2-1},
\label{eqn:sdofsolution}
\end{equation}
where $\omega_0=\sqrt{\frac{k}{m}}$ and $\xi=\frac{C}{2m\omega_0}$ represent the natural frequency and damping ratio, respectively.
As we can see, the natural frequencies and damping ratios will always be positive, and depending on the value of $\xi$, $s$ may be real or complex and there may be one or two possible solutions.

In this paper, we restrict ourselves to the case when there is no damping, i.e., $C=0$, and thus $\xi=0$.
From equation~\eqref{eqn:sdofsolution} we can see that when $\xi=0$, we have two purely imaginary solutions $s_1=i\omega_0$, and $s_2=-i\omega_0$.
Thus, $x_1(t)=A e^{i\omega_0t}$, and $x_2(t)=B e^{-i\omega_0t}$ are both eligible solutions to equation~\eqref{eqn:sdof}.
In fact, any solution to the above equation can be expressed as a linear combination of $x_1(t)$ and $x_2(t)$, such that the general solution can be written as $x(t)=A e^{i\omega_0 t}+B e^{-i\omega_0 t} =(A+B)\cos(\omega_0 t)+i(A-B)\sin(\omega_0 t)$.
%
%
Furthermore, noting that $x(0)=A+B$ and $\dot{x}(0)=i (A-B)\omega_0$, it follows that $x(t)=x(0)\cos(\omega_0 t)+\frac{\dot{x}(0)}{\omega_0}\sin(\omega_0 t)$.

Since we want to deal with real valued signals $x(t)$, this demands that $B=A^{*}$.
Denoting $A=a+ib$, we can once again rewrite the solution as $x(t)=2a\cos(\omega_0 t)-2b\sin(\omega_0 t)$.
Finally, because any linear combination of sines and cosines with the same frequency is also a sine wave with the same frequency, we can rewrite this as $x(t)=\rho\sin(\omega_0 t+\theta)$, where
\begin{equation}
\rho=2\sqrt{a^2+b^2} ~\text{and}~
\theta=
\left\{
\begin{array}{ll}
\arcsine\left(\frac{a}{\sqrt{a^2+b^2}}\right),& b\leq 0, \\
\pi-\arcsine\left(\frac{a}{\sqrt{a^2+b^2}}\right),& b>0.
\end{array}
\right.
\label{eqn:SDOFrhotheta}
\end{equation}

\subsection{Multiple-degree-of-freedom system}
\label{subsec:mdof}

Similar to the SDOF system, an $N$-degree MDOF system\footnote{Theoretically, a structure will have infinitely many degrees of freedom. However, the number of mode shapes that one can detect is equal to the number of sensors placed on the structure. In the following, whenever we deal with an $N$-degree MDOF system, we are implicitly assuming that we have $N$ sensor nodes deployed on the structure.} can be formulated as
\begin{equation*}
[M]\{ \ddot{u}(t)\} +[C]\{ \dot{u}(t)\}+[K]\{ u(t)\} = \{0(t)\},
\end{equation*}
where $[M]$ is an $N\times N$ diagonal mass matrix, $[C]$ is a symmetric $N\times N$ damping matrix, $[K]$ is an $N\times N$ symmetric stiffness matrix, and $\{u(t)\}$ is an $N\times 1$ vector of displacement signals.
Note that $\{u(t)\}=\{u_1(t),\dots,u_N(t)\}$, and each $\{u(t)\}(l)=u_l(t)$, $l=1,\dots,N$, is a displacement signal.
One can view $u_l(t)$ as the signal being observed at the $l$th sensor node.

Again we consider an undamped system and set $[C]=[0]$. This simplifies the above equation to $[M]\{ \ddot{u}(t)\} + [K]\{ u(t)\} = \{0(t)\}$.
Let us assume $\{u(t)\}=\{\psi\}Ae^{i\omega t}$ to be a solution to this equation, where $\omega,t\in\mathbb{R}$. Here, $\{\psi\}$ is an $N\times 1$ spatial vector that is independent of time; we define it to have unit energy, i.e., $\|\{\psi\}\|_2=1$ (we can assume this as the normalization can be absorbed into the scalar variable $A$).
Plugging in the appropriate derivative to the above expression we get $\left(-[M]\omega^2 + [K]\right)\{\psi\}Ae^{i\omega t} = \{0(t)\}$.
Since this must hold for all values of $t$ it must be that
\begin{equation}
\left([K]-\omega^2[M]\right)\{\psi\} = \{0\}.
\label{eq:eigKpsi}
\end{equation}
The above represents a generalized eigenvalue problem and our objective is to find pairs of $\omega^2$ and $\{\psi\}$ that satisfy this equation.
Notice the similarity to the conventional eigenvalue problem, which corresponds to the case when $[M]=I$.
To solve the above problem, as in the conventional eigenvalue problem, one starts off by computing the generalized eigenvalue $\omega^2$ that satisfies
\begin{equation}
\text{det}\left([K]-\omega^2[M]\right)=0.
\label{eq:eigdet}
\end{equation}
Assuming that this does not vanish as a function of $\omega^2$ and that $M$ is full rank, the left hand side of \eqref{eq:eigdet} will represent an $N$th order polynomial and there will be $N$ generalized eigenvalues $\omega^2$ as roots of this polynomial.
Each $\omega^2$ when plugged back in to \eqref{eq:eigKpsi} will have a corresponding generalized eigenvector $\{\psi\}$.
Thus, there will be $N$ generalized eigenvalues $\omega_1,\dots,\omega_N$, which are known as the modal frequencies, and $N$ corresponding generalized eigenvectors $\{\psi_1\},\dots,\{\psi_N\}$, which are known as the mode shape vectors.
Each modal frequency will be real and positive, and the mode shape vectors will be orthonormal to one another.
Without loss of generality, we assume the frequencies are sorted such that $\omega_1\geq\omega_2\geq\dots\geq\omega_N>0$.

It is clear that each $\{\psi_n\}A_ne^{i\omega_n t}$, $n=1,\dots,N$, will be a valid solution to the MDOF system equation.
Furthermore, as in the SDOF system $\{\psi_n\}B_ne^{-i\omega_n t}$ will also be a valid solution and thus for each $n$ a complete solution will be of the form $\{\psi_n\}\left(A_n e^{i\omega_n t} + B_ne^{-i\omega_n t}\right)$.
We can guarantee this solution is real by ensuring that $\{\psi_n\}$ is real (there exists a real eigenvector that satisfies the above equation given that the mass and stiffness matrices are real and symmetric) and $B_n=A_n^{*}$.
Thus, as with the SDOF case we can rewrite this solution as $\{\psi_n\}\rho_n\sin(\omega_n t +\theta_n)$, where $\rho_n$ and $\theta_n$ are as defined in~\eqref{eqn:SDOFrhotheta}.
Finally, it is easy to see that all linear combinations of this solution are valid solutions to the MDOF system equation and thus the general solution is of the form
\begin{equation}
\{u(t)\}=\sum_{n=1}^{N}\{\psi_n\}\rho_n\sin(\omega_n t +\theta_n).
\label{eqn:modalsuperposition}
\end{equation}
In the structural dynamics community this is known as the modal superposition equation.

\section{Problem Formulation}

\subsection{The analytic signal of $\{u(t)\}$}
\label{sec:analytic}

Our framework and analysis will involve sampling what is known as the analytic signal of $\{u(t)\}$~\cite{boualem03}.
\begin{definition}[Definition 1.2.1,~\cite{boualem03}]
A signal $v(t)$ is said to be {\em analytic} iff $V(f)=0~\text{for}~f<0$, where $V(f)$ is the Fourier transform of $v(t)$.
\end{definition}
An analytic signal can be obtained by removing all negative frequencies in a given signal.
The analytic signal is a frequently used representation in mathematics, signal processing, and communications; in some problems (such as ours) it can simplify the mathematical manipulations.

To discuss the analytic signal of $\{u(t)\}$, let us examine each entry in $\{u(t)\}$, i.e., $u_l(t)$. Based on the derivation in Section~\ref{subsec:mdof}, each $u_l(t)$ can be written as
\begin{align*}
u_l(t) 
&=\sum_{n=1}^{N} \{\psi_n\}(l)(A_n e^{i\omega_nt}+A_n^{*} e^{-i\omega_nt}).
\end{align*}
Thus, the analytic signal of $u_l(t)$, represented as $v_l(t)$, is simply
$$
v_l(t)=\sum_{n=1}^{N} \{\psi_n\}(l)A_n e^{i\omega_nt},
$$
and the analytic signal of the entire vector $\{u(t)\}$, denoted as $\{v(t)\}$, can be written as
\begin{equation}
\{v(t)\}=\sum_{n=1}^{N} \{\psi_n\}A_n e^{i\omega_nt}.
\label{eq:modsupeqn}
\end{equation}
Note that $\{v(t)\}$ is no longer real but complex.

Obtaining an analytic signal in practice involves the application of a Hilbert transform.
However, detailed discussion of this matter is out of scope of this paper and we will refer interested readers to~\cite{frederick09} for more information.
For the remainder of this paper, we will assume that we have successfully extracted the analytic signal from each $u_l(t)$.
Thus, all derivations from here onwards will be in terms of $\{v(t)\}$.

\subsection{The relationship to the SVD}

We can write the modal superposition equation \eqref{eq:modsupeqn} in matrix-vector multiplication format as
\begin{equation*}
\{v(t)\} = [\Psi][\Gamma]\{s(t)\}, 
\end{equation*}
where $[\Psi] = [\{\psi_1\}, \{\psi_2\},\dots,\{\psi_N\} ]$ denotes the $N\times N$ mode shape vector matrix, which as mentioned before, has orthonormal columns,
$$
[\Gamma] = \sqrt{M} \cdot \left[
\begin{array}{cccc}
A_1 & 0 & \dots & 0 \\
0 & A_2 & \dots & 0 \\
\vdots & \vdots & \ddots & \vdots \\
0 & 0 & \dots & A_N \\
\end{array}
\right ]
$$
denotes an $N\times N$ diagonal matrix, and
$$
\{s(t)\} = \frac{1}{\sqrt{M}} \cdot \left\{
\begin{array}{c}
e^{i\omega_1t}\\
e^{i\omega_2t}\\
\vdots \\
e^{i\omega_Nt}\\
\end{array}
\right\}
$$
denotes an $N\times 1$ modal coordinate vector.

In order to see how the SVD could be useful in extracting the modal parameters, let us suppose that we sample each row of $\{v(t)\}$ at $M$ distinct points in time $t_1,\dots,t_M$.
We assume $M \ge N$ and denote the resulting $N\times M$ data matrix as
\begin{equation}
[V]=\left[
\begin{array}{cccc}
v_1(t_1) & v_1(t_2)&\dots & v_1(t_M) \\
v_2(t_1) & v_2(t_2)&\dots & v_2(t_M) \\
\vdots \\
v_N(t_1) & v_N(t_2)&\dots & v_N(t_M) \\
\end{array}
\right]
\in\mathbb{C}^{N\times M}.
\label{eq:V}
\end{equation}
The sampling of $\{v(t)\}$ at $t_1,\dots,t_M$ implies the sampling of $\{s(t)\}$ at the exact same time points which leads us to define
\begin{equation*}
[S]= \frac{1}{\sqrt{M}} \cdot \left[
\begin{array}{cccc}
e^{i\omega_1t_1} & e^{i\omega_1t_2}&\dots & e^{i\omega_1t_M} \\
e^{i\omega_2t_1} & e^{i\omega_2t_2}&\dots & e^{i\omega_2t_M}\\
\vdots \\
e^{i\omega_Nt_1} & e^{i\omega_Nt_2}&\dots & e^{i\omega_Nt_M}
\end{array}
\right]
\in\mathbb{C}^{N\times M}
\end{equation*}
and allows us to write the matrix of samples as
\begin{equation}
[V]=[\Psi][\Gamma][S].
\label{eqn:svdofv}
\end{equation}

Equation \eqref{eqn:svdofv} makes explicit the relationship between the SVD and the modal parameters.
We know that $[\Psi]$ is a square matrix with orthonormal columns, and $[\Gamma]$ is a diagonal matrix.
Hypothetically, if $[S]$ happened to be a matrix with orthogonal (or orthonormal) rows, then equation~\eqref{eqn:svdofv} would precisely describe the SVD of $[V]$.
In that case, one could obtain the modal parameters by simply computing the SVD of $[V]$!
%

%
As an example, the rows of $[S]$ would be perfectly orthogonal if they happened to equal $N$ distinct length-$M$ DFT vectors.
This could be ensured if we sampled at uniform times $t_m=T_s (m-1)$, where $T_s$ is a sampling interval, $m\in\{1,\dots,M\}$, and $M \ge N$, but it would require the modal frequencies to lie on a grid such that $\omega_n=\frac{2\pi k_n}{MT_{s}}$, where $k_n\in\{1,\dots,M\}$.
If these conditions were satisfied, the SVD of $[V]$ would exactly recover the modal parameters.
Unfortunately, this is an unrealistic model for the purpose of modal analysis because the modal frequencies will typically not lie on a grid.

If we drop the assumption that the modal frequencies lie on a grid, the problem becomes much more complicated, and in general the rows of $[S]$ will not be orthogonal.
However, our main results, presented in Section~\ref{sec:ourmainresults}, rely on characterizing sampling strategies that ensure the rows of $[S]$ will still be nearly orthogonal and showing in these situations that the mode shape vectors can be accurately estimated by computing the SVD of $[V]$.

\subsection{SVD in modal analysis}

Among the many techniques that have been proposed for modal analysis, we briefly mention a few that also make use of the SVD.
The Ibrahim Time Domain (ITD)~\cite{itd1977} method shares a number of similarities with our proposed method. This method also begins with the MDOF model and sets up an equation that relates the modal parameters to the observations akin to~\eqref{eqn:modalsuperposition}. After further algebraic manipulations, a set of equations reveals that one can extract the modal parameters via an eigendecompsition. Based on this observation, the ITD method obtains estimates of the modal parameters by computing the eigendecomposition of a matrix that is a function of the observed data matrix.
In the Frequency Domain Decomposition (FDD)~\cite{Brincker_Zhang_Andersen_2000} method, the observations are used to compute cross power spectral density estimates. These cross power spectral density estimates are collectively a 3D data cube that consists of a 2D cross spectral matrix at each frequency. Given these estimates, the SVD is used to extract the singular vectors of the 2D cross spectral matrices. These singular vectors provide estimates of the mode shape vectors.
The Eigensystem Realization Algorithm (ERA)~\cite{era1985} method makes use of ideas in control theory and sets up the problem with a state-space equation of an MDOF system. The main use of the SVD in this method is to decompose the Hankel matrix that can be constructed from measured impulse response data. The singular vectors and singular values of the Hankel matrix are then manipulated in order to form a system matrix describing the underlying system. The estimates of the modal parameters are then computed via an eigendecomposition on this system matrix.

All of these methods implicitly assume that the observations are sampled at uniform intervals at or above the Nyquist rate. Thus, these methods may not be directly applicable when the observations are sampled in a random fashion. Furthermore, to the best of our knowledge none of these methods are accompanied by error analysis or instructions on how long to sample the vibration signal.
In the next section, we present our proposed method along with detailed analysis providing error bounds and sufficient conditions on how to sample in order to guarantee faithful recovery of the mode shape vectors.

\section{Main Results}
\label{sec:ourmainresults}

In this section, we present our main results.
We propose three measurement schemes---uniform time sampling, random time sampling, and uniform time sampling followed by a random matrix multiplication---and for each measurement scheme we provide a sufficient condition for the accurate recovery of mode shapes via the SVD.
Proofs of all of our results appear in the Appendix.

The main focus of this paper is the recovery of the mode shape vectors.
However, in Section~\ref{csshm:sec:shmexp} we do provide a short discussion along with promising simulation results concerning the recovery of the modal frequencies.

In the results that follow, we use $\delta_{\min}$ and $\delta_{\max}$ to denote lower and upper bounds on the minimum and maximum separation of the modal frequencies.
In other words, we assume that $\delta_{\min}\leq \min_{l\neq n}|\omega_l-\omega_n|$ and $\delta_{\max}\geq \max_{l\neq n}|\omega_l-\omega_n|$.
Furthermore, we use $t_{\max}$ to denote the total sampling time span.
Finally for $a,b \in [0,1]$, the quantity $D(a||b):=a(\log(a)-\log(b))+(1-a)(\log(1-a)-\log(1-b))$ is known as the binary information divergence, or the Kullback-Leibler divergence~\cite{JTtailbounds}.

\subsection{Uniform time sampling and random time sampling}

Our proposed method for recovering the mode shape vectors from uniform time samples or from random time samples is very simple and is described in Algorithm~\ref{alg:algorithm}.
\begin{algorithm}[t]
	\KwIn{Data matrix $[V]$ as defined in \eqref{eq:V}}
	\KwOut{$\hat{[\Psi]}$ (left singular vectors of $[V]$)}
	${\rm SVD([V])}=[\hat{\Psi}][\hat{\Gamma}][\hat{S}]$ \\
\caption{Pseudo-code for mode shape estimation}
\label{alg:algorithm}
\end{algorithm}
In particular, our method simply computes the SVD of $[V]$ and returns the matrix of left singular vectors $[\hat{\Psi}]=[\{\hat{\psi}_1\},\dots,\{\hat{\psi}_N\}]$ as estimates of the true mode shape matrix $[\Psi]$.
One point to note about this algorithm is that because $[V]$ is $N\times M$, where we assume $M\geq N$, the dimensions of $[\hat{\Psi}]$, $[\hat{\Gamma}]$, and $[\hat{S}]$ will be $N\times N$, $N\times M$, and $M\times M$, respectively. These differ from the dimensions of their respective counterparts in~\eqref{eqn:svdofv}. Taking a closer look and noting that only $N$ diagonal entries in $[\hat{\Gamma}]$ are non-zero, we can compute the truncated SVD to obtain the desired dimensions.

\subsubsection{Uniform time sampling}
\label{subsec:uniformsampling}

Let us now suppose that we sample $\{v(t)\}$ uniformly in time with a uniform sampling interval denoted by $T_s$.
The sampling times are given by $t_m=(m-1)T_s$, $m=1,\dots,M$.
We are therefore sampling within the time span $[0,t_{\max}]$, where $t_{\max}:=(M-1)T_s$.
We can establish the following theorem.

\begin{theorem}
\label{thm:uniformsampling}
Let $[V]=[\Psi][\Gamma][S]$ be as given in~\eqref{eqn:svdofv} describing an $N$-degree-of-freedom system sampled according to the uniform sampling scheme described above.
For $0 < \epsilon < 1$, suppose we sample for a total time span of at least
\begin{equation}
t_{\max}\geq \frac{2\pi(\log\lfloor N/2\rfloor + 1.01)}{\epsilon\delta_{\min}}
\label{eq:tmaxuniform}
\end{equation}
with sampling interval
$
T_s=\frac{\pi}{\delta_{\max}}
$
and ensure that $M \ge N$.
Or, equivalently, suppose we take
\begin{equation}
M\geq \max\left( \frac{2(\log\lfloor N/2 \rfloor+1.01)}{\epsilon}\frac{\delta_{\max}}{\delta_{\min}}+1 ,\; N \right)
\label{eq:Muniform}
\end{equation}
total samples with the sampling interval
$
T_s=\frac{\pi}{\delta_{\max}}.
$
Then, the mode shape estimates $[\hat{\Psi}]$ obtained via Algorithm~\ref{alg:algorithm} satisfy
\begin{equation}
\|\{\psi_n\}-\{\hat{\psi}_n\}\|_2\leq \min\left\{\sqrt{2}, ~ \frac{\epsilon\sqrt{1+\epsilon}}{\sqrt{1-\epsilon}} \cdot \mathrm{sep}_n(\epsilon) \right\},
\label{eqn:resultuniform}
\end{equation}
where
\begin{equation*}
\mathrm{sep}_n(\epsilon) = \max_{l\neq n}\frac{\sqrt{2}|A_l||A_n|}{\displaystyle\min_{c\in[-1,1]}\{||A_l|^2-|A_n|^2(1+c\epsilon) |\}  }.
\end{equation*}
\end{theorem}

In \eqref{eqn:resultuniform} we see that the error in the $n$th estimated mode shape vector mainly depends on $\epsilon$ and what is essentially the minimum separation between $|A_n|$ and all other $|A_l|$.
The variable $\epsilon$ controls how close the rows of $[S]$ are to being orthogonal; a small choice of $\epsilon$ implies more orthogonal rows and leads to a better preservation of the mode shapes but requires more samples.
Furthermore, the bigger the separation between $|A_n|$ and all other $|A_l|$, the better our estimate.
Note that the parameters $|A_n|$ are dependent on the underlying structure and thus are out of our control.
In order to guarantee a small error in the $n$th mode shape when $|A_n|$ is close to some other $|A_l|$, one would need to make $\epsilon$ smaller.

Turning our attention to the sampling parameters, the above theorem essentially tells us that we need to sample for a time span that is inversely proportional to the minimum spacing between the modal frequencies.
Thus, the smaller the minimum spacing between the modal frequencies, the longer we must sample to get an accurate estimate.
Also, since $T_s=\frac{\pi}{\delta_{\max}}$, the maximum spacing between modal frequencies determines how fast we need to sample.
Comparing this sampling interval to the Nyquist sampling interval which would be $T_{0}=\frac{\pi}{2\max_{n}\omega_n}$, it is interesting to note that $T_s> T_0$.
This suggests that for the purpose of mode shape extraction, we can potentially sample at a rate lower than the Nyquist rate and still accurately recover the mode shapes.
However, it is important to note that in order to sample with $T_s$ we must know in advance the maximum separation between modal frequencies.
In scenarios where $\delta_{\max}$ is unknown it would be more reasonable to sample at a sufficiently small interval $(T_s \approx T_0)$ to ensure the sampling conditions are satisfied.
Finally, note that the condition on $M$ is fairly satisfactory in its logarithmic dependence on $N$ and its linear dependence on $\frac{1}{\epsilon}$ (this assumes the left hand term dominates in \eqref{eq:Muniform}).
However, it also scales with the ratio $\frac{\delta_{\max}}{\delta_{\min}}$, which depends on the structure.
For some structures this ratio could in fact be large, and in the absence of additional information about the structure, one may need to assume this ratio is large.
This motivates our second sampling strategy, which appears below.

\subsubsection{Random time sampling}
\label{subsec:randomsampling}

Let us now suppose that we sample $\{v(t)\}$ at $M$ time points $t_1,\dots,t_{M}$ chosen uniformly at random in the time interval $[0,t_{\max}]$.
We can establish the following theorem.
\begin{theorem}
Let $[V]=[\Psi][\Gamma][S]$ be as given in~\eqref{eqn:svdofv} describing an $N$-degree-of-freedom system sampled according to the random sampling scheme described above.
Suppose we sample for a total time span of at least
\begin{equation}
t_{\max}\geq\frac{40(\log\lfloor N/2\rfloor+1.01)}{\epsilon\delta_{\min}},
\label{eq:tmaxrandom}
\end{equation}
and within this time span suppose we take a number of measurements satisfying
\begin{equation}
M > \max\left( \frac{\log(N) + \log(2/\tau)}{ \min(D_1,D_2) },\; N \right)
\label{eq:Mrandom}
\end{equation}
where
\begin{align*}
D_1 &= D((1+\epsilon)/N ||(1+\epsilon/10) /N),\\
D_2 &= D((1-\epsilon)/N ||(1-\epsilon/10) /N).
\end{align*}
Then with probability at least $1-\tau$ all of the mode shape estimates $[\hat{\Psi}]$ obtained via Algorithm~\ref{alg:algorithm} will satisfy \eqref{eqn:resultuniform}.
\label{thm:randomsamplingFinal}
\end{theorem}

This result for random time sampling looks somewhat similar to Theorem~\ref{thm:uniformsampling} for uniform time sampling.
The recovery guarantee is the same and the required time span differs only by a constant.
A critical difference, however, is that the requisite number of samples $M$ no longer depends on the ratio $\frac{\delta_{\max}}{\delta_{\min}}$.
However, when $N$ is large and $\epsilon$ is small, the denominator in~\eqref{eq:Mrandom} will scale like $\frac{\epsilon^2}{N}$, and so the requisite number of samples will scale like $\frac{N\log(N)}{\epsilon^2}$.
This represents a stronger dependence on $N$ compared to what appears in~\eqref{eq:Muniform}, but only by a logarithmic factor (because the right hand term in~\eqref{eq:Muniform} scales like $N$).
It also represents a stronger dependence on $\epsilon$ compared to what appears in~\eqref{eq:Muniform}.
Ultimately, we see that in some cases random time sampling could provide a significant reduction in the number of samples for systems where $\frac{\delta_{\max}}{\delta_{\min}}$ is large or unknown.
For a given problem, the better choice between uniform and random time sampling may depend on the particular circumstances and the parameters of the system under study.

\subsection{Uniform sampling followed by random matrix multiplication}

The last measurement scheme that we consider involves taking uniform time samples and compressing these via multiplication by a random matrix.
More specifically, let us form $[V]$ following the uniform sampling scheme as discussed in Section~\ref{subsec:uniformsampling} with $t_{\max}$ and $T_s$ as given in Theorem~\ref{thm:uniformsampling}.
Subsequently, we construct a random $M\times M'$ matrix $[\Phi]$ and compute the $N \times M'$ matrix $[Y]=[V][\Phi]$ of compressed measurements.
We are specifically interested in cases where $M' < M$, i.e., when $[Y]$ has fewer columns than $[V]$.
We also note that $[\Phi]$ can be applied individually to each row of $[V]$ and the resulting measurements can be concatenated to form $[Y]$.
This means that this CS measurement scheme can be performed sensor-by-sensor in a SHM system.

To state our results, we write the truncated SVD of $[Y]$ analogously to that of $[V]$ as $[Y]=[\tilde{\Psi}][\tilde{\Gamma}][\tilde{S}]$. The matrices $[\tilde{\Psi}]=[\{\tilde{\psi}_1\},\dots,\{\tilde{\psi}_N\}]$, $[\tilde{\Gamma}]$ and $[\tilde{S}]$ will be $N\times N$, $N\times M'$, and $M'\times M'$, respectively.
We also require the following definition.
\begin{definition}
\label{def:jl}
A $Q \times m$ random matrix $[\Phi]$ is said to satisfy the {\em distributional JL property} if for any fixed $\{x\} \in \complex^Q$ and any $0 < \epsilon < 1$,
$$
\text{Pr}\left[ \left| \| [\Phi]^* \{x\} \|_2^2 - \| \{x\} \|_2^2 \right| > \epsilon \| \{x\} \|_2^2 \right] \le 4 e^{-m f(\epsilon)},
$$
where $f(\epsilon) > 0$ is a constant depending only on $\epsilon$.
\end{definition}
For most random matrices satisfying the distributional JL property, the functional dependence on $\epsilon$, $f(\epsilon)$, is quadratic in $\epsilon$ as $\epsilon \rightarrow 0$.
There are a variety of random matrix constructions known to possess the distributional JL property. Notably, random matrices populated with independent and identically distributed (i.i.d.)\ subgaussian entries will possess this property~\cite{Daven_Concentration}. Subgaussian distributions include suitably scaled Gaussian and $\pm 1$ Bernoulli random variables.

We are now ready to state our next theorem.
\begin{theorem}
\label{thm:uniformandrandom}
Let $[V]=[\Psi][\Gamma][S]$ be as given in~\eqref{eqn:svdofv} describing an $N$-degree-of-freedom system sampled according to the uniform sampling scheme with $T_s$ and $t_{\max}$ as required by Theorem~\ref{thm:uniformsampling}.
Let $[\Phi]$ represent an $M\times M'$ random matrix that satisfies the distributional JL property with
$$
M'\geq\frac{2k\log(42/\epsilon')+\log(4/\delta)}{f(\epsilon'/\sqrt{2})},
$$
where $k\leq N$ represents the rank of $[V]$ and $\epsilon'$ represents the distortion factor of $[\Phi]$.
Let $[Y]=[V][\Phi]$ and let $[\tilde{\Gamma}]$ and $[\tilde{\Psi}]$ denote the estimated singular values and left singular vectors of $[Y]$ returned by Algorithm~\ref{alg:algorithm} when we provide $[Y]$ in place of $[V]$ as the input matrix.
Then, with probability exceeding $1-\delta$ the mode shape estimates $[\tilde{\Psi}]$ satisfy the following bound:
\begin{align}
& \|\{\psi_n\}-\{\tilde{\psi}_n\}\|_2\leq \nonumber \\
& ~~~ \min\left\{\sqrt{2}, ~ \frac{\epsilon\sqrt{1+\epsilon}}{\sqrt{1-\epsilon}} \cdot \mathrm{sep}_n(\epsilon) + \frac{\epsilon'\sqrt{1+\epsilon'}}{\sqrt{1-\epsilon'}} \cdot \mathrm{sep}^{'}_n(\epsilon') \right\},
\label{eqn:resultuniformandrandom}
\end{align}
where
$$
\mathrm{sep}^{'}_n(\epsilon') = \max_{l\neq n}\frac{\sqrt{2}\sigma_l\sigma_n}{\displaystyle\min_{c\in[-1,1]}\{|\sigma_l^2-\sigma_n^2(1+c\epsilon') |\}  }
$$
and $\sigma_n$ are the singular values of $[V]$.
\end{theorem}

The error bound \eqref{eqn:resultuniformandrandom} looks similar to those appearing in Theorems~\ref{thm:uniformsampling} and~\ref{thm:randomsamplingFinal} except that instead of having one term we now have two terms.
The first term is essentially the error due to the uniform time sampling matrix $[V]$ and the second term is the error due to the multiplication by a random matrix $[\Phi]$.
The required number of columns in the random matrix $[\Phi]$ is dependent on the rank $k$ of $[V]$.
The higher the rank of $[V]$, the more measurements we need; in the worse case one could assume $k = N$.
One could easily envision a scenario where this measurement scheme could be useful.
For example, suppose we are dealing with a structure that has high $\delta_{\max}$ and small $\delta_{\min}$ (or suppose that we do not know these quantities and so we conservatively suppose they are large and small, respectively).
This means that we need to take a large number $M$ of uniform samples using a small sampling interval $T_s$ over a long time duration $t_{\max}$.
In such scenarios, one could choose to post-process the signals to reduce the number of measurements by multiplying each sample vector with $[\Phi]$.
Another similar scenario is when we have a conventional uniform sensor over which we do not have control over the sampling interval.
Again, we may use $[\Phi]$ as a way to reduce the number of measurements.
In Section~\ref{csshm:sec:subsec:synthetic}, we demonstrate the usefulness of this sampling scheme with a synthetic example.

\section{Experimental Results}
\label{csshm:sec:shmexp}

\subsection{Experiments with synthetic data}
\label{csshm:sec:subsec:synthetic}

\begin{figure*}[t]
\centerline{
\begin{minipage}[t]{0.33\linewidth}
  \centering
  \centerline{\includegraphics[width=7cm]{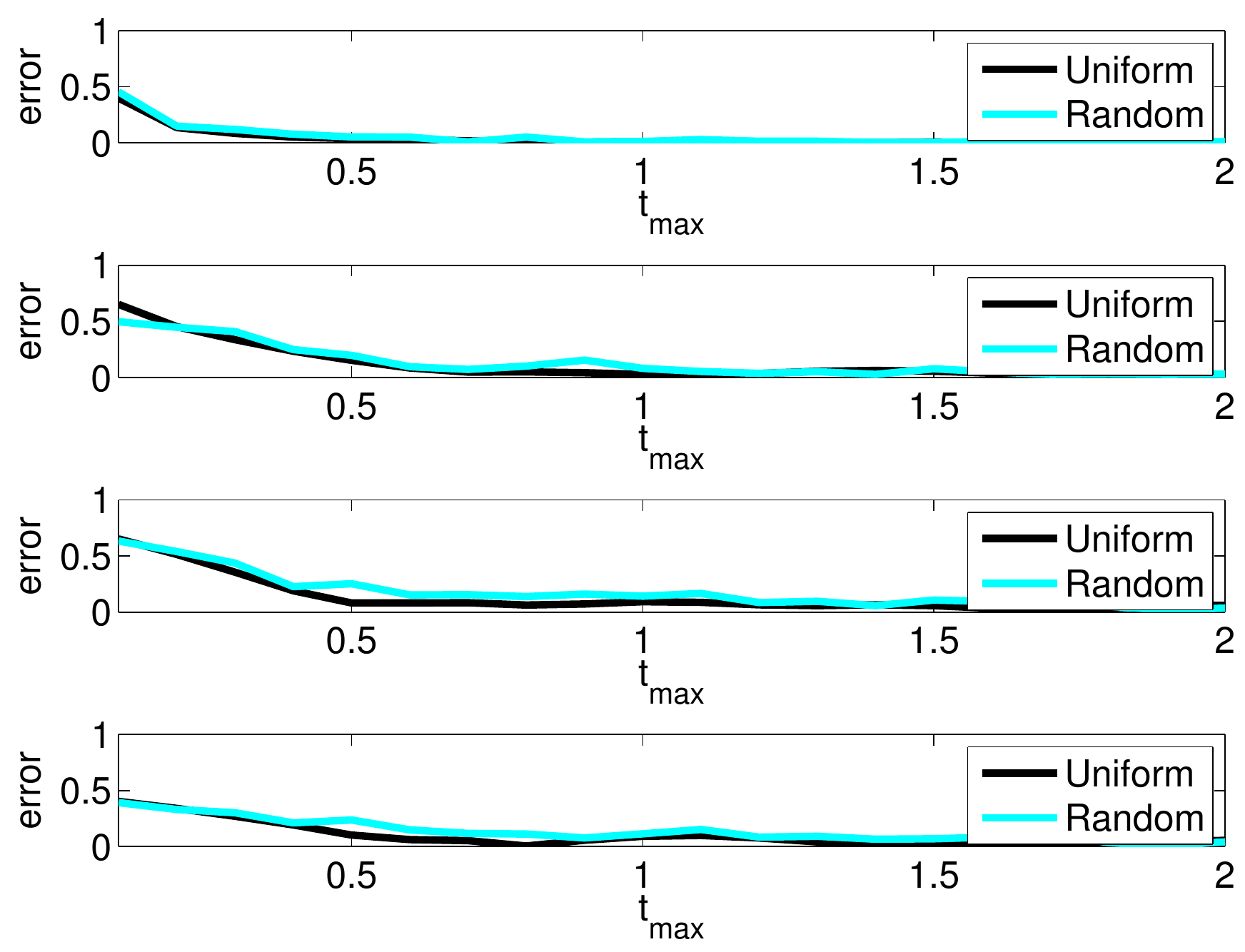}}
  \centerline{\small{(a)}}\medskip
\end{minipage}
\hfil
\begin{minipage}[t]{0.33\linewidth}
  \centering
  \centerline{\includegraphics[width=7cm]{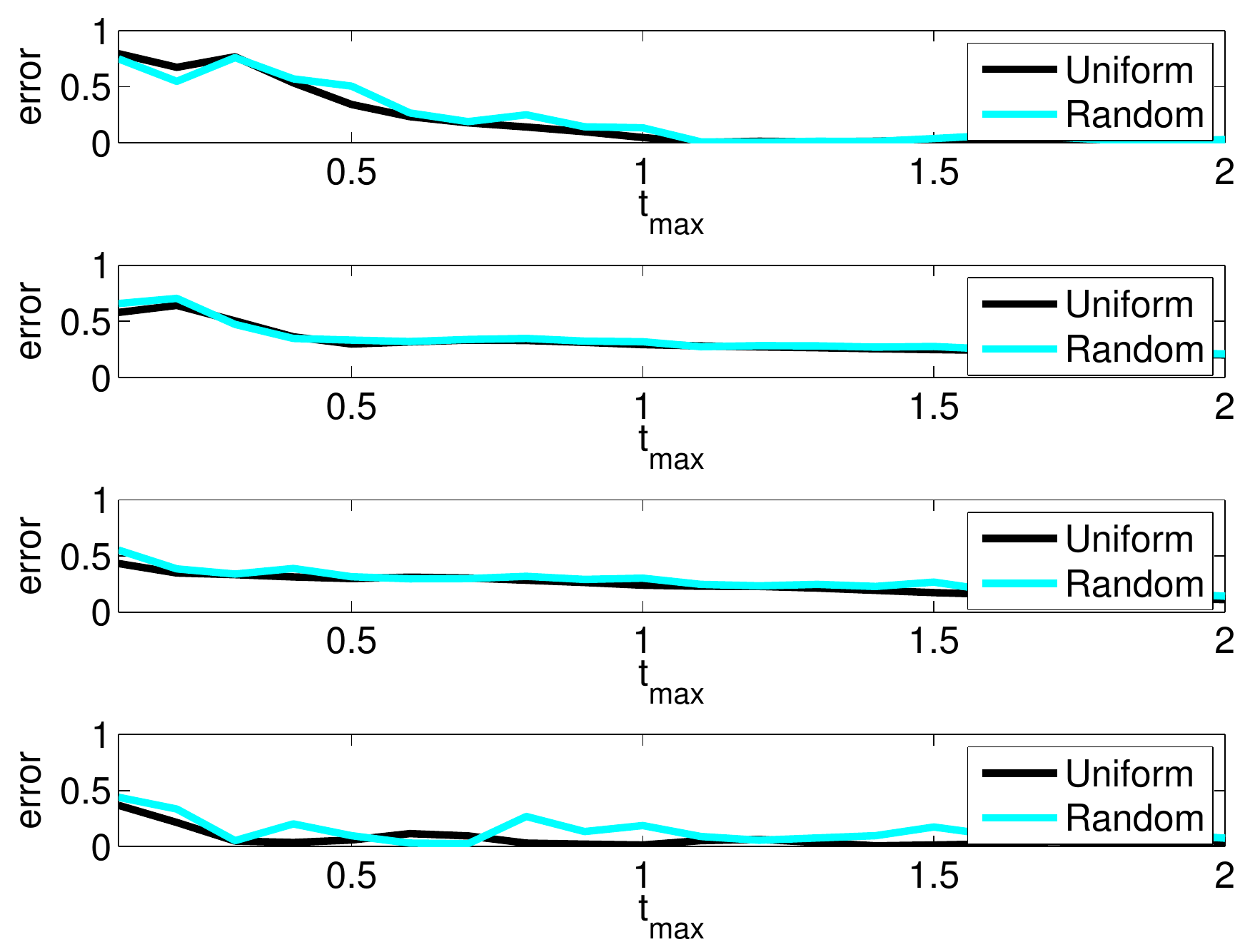}}
  \centerline{\small{(b)}}\medskip
\end{minipage}}
\caption{\small\sl {$\ell_2$ error of mode shape estimates using uniform (black) and random (blue) time sampling schemes.
The four subfigures correspond to the four mode shapes, and the results are plotted as a function of the total sampling time $t_{\max}$.
(a) Results using an inital set of modal frequencies.
(b) Results using a set of modal frequencies with smaller minimum separation $\delta_{\min}$.
}%
}
\label{fig:usrs1}
\end{figure*}

We begin by demonstrating the effectiveness of the various methods using an idealized synthetic dataset.
The system that we consider is a 4-degree-of-freedom structure with no damping and under free vibration.
For reproducibility the modal vectors are the eigenvectors of the following symmetric matrix:
$$
[Z]=\left[
 \begin{array}{cccc}
 2 & -1 & 0 & 0 \\
0 & 2 & -1 & 0 \\
0 & -1&  2 & -1  \\
0 & 0 & -1 &   2 \\
 \end{array}
 \right].
$$
We use these mode shapes throughout this subsection.
We also use the following $[\Gamma]$ matrix throughout this subsection:
$$
[\Gamma]
=
\left[
 \begin{array}{cccc}
 1 & 0 & 0 & 0 \\
0 & 0.45 & 0 & 0 \\
0 & 0 &  0.15 & 0  \\
0 & 0 & 0 &  0.01 \\
 \end{array}
 \right].
$$
The modal superposition equation can be written as
$$
\{v(t)\}=
[\Psi]
 \underbrace{
\left[
 \begin{array}{cccc}
 1 & 0 & 0 & 0 \\
0 & 0.45 & 0 & 0 \\
0 & 0 &  0.15 & 0  \\
0 & 0 & 0 &  0.01 \\
 \end{array}
 \right]}_{[\Gamma]}
 \underbrace{
\left\{
\begin{array}{c}
e^{i\omega_1t}\\
e^{i\omega_2t}\\
e^{i\omega_3t}\\
e^{i\omega_4t}\\
\end{array}
\right\}}_{\{s(t)\}},
$$
where the modal frequencies $\omega_1,\dots,\omega_4$ remain to be chosen.

In our first experiment, we demonstrate the uniform and random time sampling methods by plotting the errors of each of the four estimated mode shapes obtained from the SVD of the sampled matrix $[V]$.
For both methods, we set $\omega_1=2.1\pi$, $\omega_2=4.28\pi$, $\omega_3=6.02\pi$, and $\omega_4=8.24\pi$ rad/s.
For the uniform time sampling scheme, we use a sampling interval of $T_s=0.1$s, which is just slightly faster than what our theorem prescribes.
Using this fixed rate, we collect samples over a total time span of duration $t_{\max}$, and we repeat the experiment for $t_{\max}=[0:T_s:2]$s (the value of $M$ therefore increases with $t_{\max}$).
For each value of $t_{\max}$ and for $n = 1,\dots,4$, we plot the $\ell_2$ error $\|\{\psi_n\}-\{\hat{\psi}_n\}\|_2$ between the ground truth mode shape vector $\{\psi_n\}$ and the corresponding estimated vector $\{\hat{\psi}_n\}$ produced using the SVD on the data matrix $[V]$.
The results are shown in the black curves in Figure~\ref{fig:usrs1}(a).
For the random time sampling scheme, we use the same values of $t_{\max}$ and sample $\{v(t)\}$ uniformly at random within the interval $[0,~ t_{\max}]$.
For each value of $t_{\max}$, the total number of samples $M$ we collect is chosen to equal the corresponding value of $M$ used for uniform time sampling above.
The errors of the mode shape estimates are shown in the blue curves in Figure~\ref{fig:usrs1}(a).
We see that overall, the performance of the two sampling schemes is comparable.
This is in agreement with Theorems~\ref{thm:uniformsampling} and~\ref{thm:randomsamplingFinal}, as they both suggest the same reconstruction guarantees given that we satisfy the sampling conditions.

Our second experiment highlights the role played by $\delta_{\min}$ (the minimum separation of the modal frequencies).
We consider two sets of modal frequencies. The first set is the same one used in our first experiment; for this set $\delta_{\min} = 1.74\pi$, and recall the results plotted in Figure~\ref{fig:usrs1}(a).
For our second set we use $\omega_1=2.1\pi$, $\omega_2=4.28\pi$, $\omega_3=4.6\pi$, and $\omega_4=8.24\pi$ rad/s.
This set has a smaller minimum separation between the modal frequencies; in particular, $\delta_{\min} = 0.32\pi$.
Based on Theorems~\ref{thm:uniformsampling} and~\ref{thm:randomsamplingFinal}, we anticipate the need to sample for a longer time span (larger $t_{\max}$) when $\delta_{\min}$ is smaller.
The mode shape errors using the second set of modal frequencies are plotted in Figure~\ref{fig:usrs1}(b).
Comparing to the results from the first set, we see that a longer sampling duration $t_{\max}$ is indeed needed to achieve comparable accuracy in estimating the mode shapes.

For uniform time sampling, when $T_s$ is fixed, then increasing $t_{\max}$ will automatically require more samples $M$ to be acquired.
For random time sampling, however, $t_{\max}$ and $M$ can be chosen independently of one another.
While in Figure~\ref{fig:usrs1} for each $t_{\max}$ we have always used the same value of $M$ for random sampling as we used for uniform sampling, this is not actually necessary.
In fact, our Theorem~\ref{thm:randomsamplingFinal} suggests that, for random sampling when $\delta_{\min}$ is small, we can increase $t_{\max}$ without increasing the number of samples $M$.
To demonstrate this, we conduct a third experiment, and for this we use the second (more closely spaced) set of modal frequencies above.
For several values of $M$, we collect $M$ samples using both uniform time sampling (for which $t_{\max}$ will be determined by $M$) and random time sampling.
For each value of $M$ with random time sampling, however, we choose $t_{\max}$ to be $2$ seconds longer than the value of $t_{\max}$ used for uniform sampling with the same value of $M$.
The results are shown in Figure~\ref{fig:usrs2}.
We see that simply by increasing $t_{\max}$ without affecting $M$, the random time sampling scheme can accommodate the decreased value of $\delta_{\min}$.

\begin{figure}[t]
\centerline{
\begin{minipage}[t]{0.33\linewidth}
  \centering
  \centerline{\includegraphics[width=7cm]{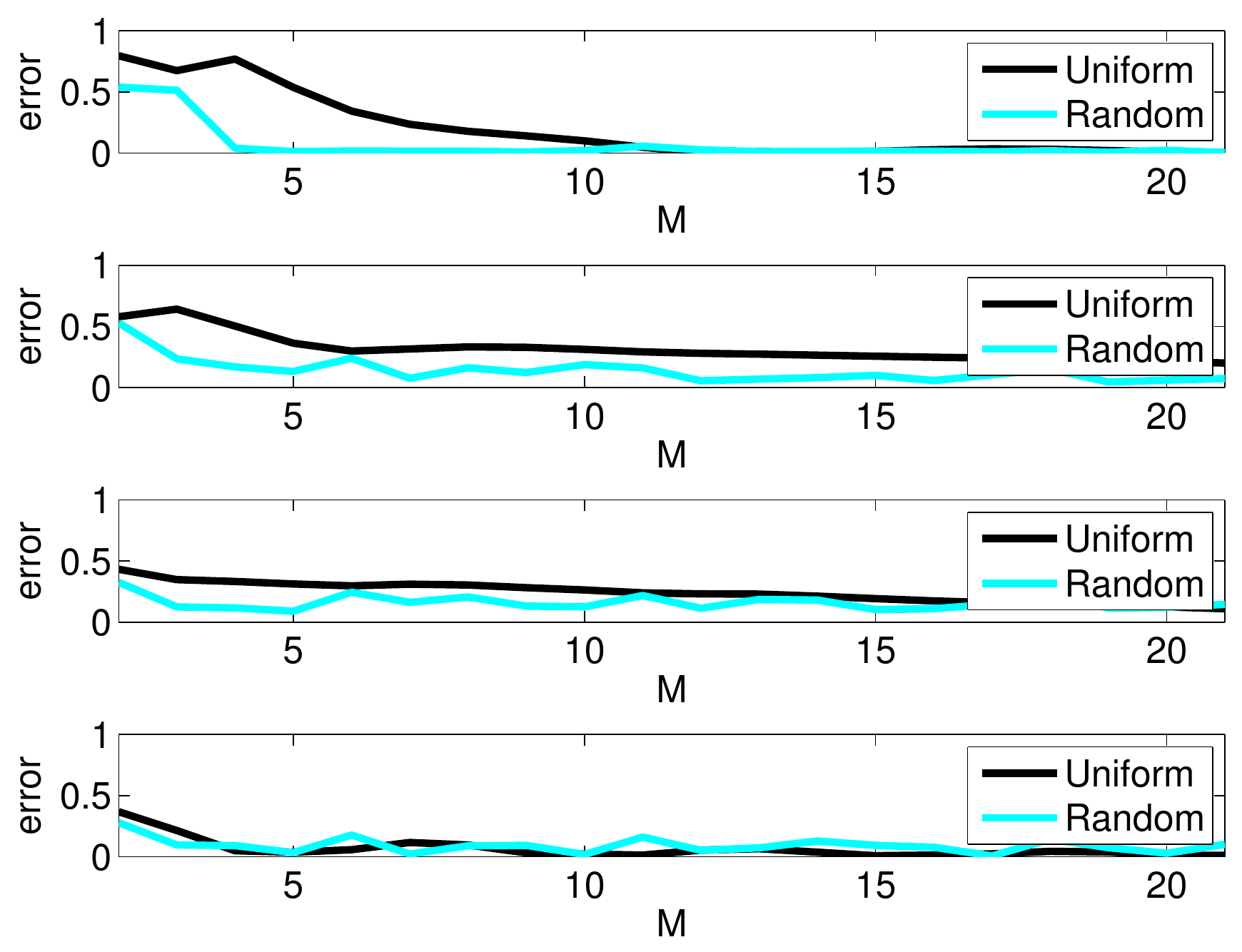}}
  \medskip
\end{minipage}
}
\caption{\small\sl {$\ell_2$ error of mode shape estimates using uniform (black) and random (blue) time sampling schemes.
The results are plotted as a function of the total number of samples $M$.
For each value of $M$, the $t_{\max}$ for random sampling is chosen to be $2$ seconds longer than the $t_{\max}$ for uniform sampling.
}%
}
\label{fig:usrs2}
\end{figure}

To motivate our fourth experiment, consider a scenario where we are limited in the number of samples we can transmit to the central data repository.
If our sensor is limited to collecting uniform time samples, 
then we may not be able to sample and transmit at a fast enough rate to avoid aliasing and accurately recover the mode shape vectors.
In a scenario such as this, one way to improve the performance would be to first sample uniformly at a high rate but then to multiply the high rate sample vector by a random compressive matrix so that the amount of transmitted data is reduced.
To illustrate this, we set the modal frequencies to be $\omega_1=10.6\pi$, $\omega_2=106.2\pi$, $\omega_3=200.8\pi$, and $\omega_4=360\pi$ rad/s, and we sample for a total time span of $t_{\max}=2$s.
For these modal frequencies, in order to avoid aliasing, the uniform sampling interval would need to satisfy $T_s\leq\frac{2\pi}{2\omega_4}=0.0028$s.
We first obtain a data matrix by sampling with $T_s=0.0629$s (a sub-Nyquist rate) over the total sampling time span $t_{\max} = 2$s.
This gives us in total 32 samples.
We then obtain a second data matrix by sampling with $T_s = 0.002$s (a super-Nyquist rate) and then multiplying the sample vector by a random Gaussian matrix $[\Phi]$ to produce $M'=32$ measurements.
We compute the left singular vectors for each of these data matrices to estimate the mode shapes.
For the first matrix, we see errors of $\|\{\psi_1\}-\{\hat{\psi}_1\}\|_2=0.39832$, $\|\{\psi_2\}-\{\hat{\psi}_2\}\|_2 = 0.71303$, $\|\{\psi_3\}-\{\hat{\psi}_3\}\|_2 = 0.58411$, and $\|\{\psi_4\}-\{\hat{\psi}_4\}\|_2 = 0.03080$.
For the second matrix, we see errors of $\|\{\psi_1\}-\{\tilde{\psi}_1\}\|_2=0.05722$, $\|\{\psi_2\}-\{\tilde{\psi}_2\}\|_2 = 0.08615$,
$\|\{\psi_3\}-\{\tilde{\psi}_3\}\|_2 = 0.03706$, and $\|\{\psi_4\}-\{\tilde{\psi}_4\}\|_2= 0.00351$.
These results illustrate the utility of random matrix multiplication for reducing the dimensionality of a uniform sample vector. 

As a final experiment on the synthetic data, we demonstrate a simple way to estimate the modal frequencies from the data matrix.
We consider a uniform sampling scenario and set the modal frequencies as $\omega_1=6.24\pi$, $\omega_2=20.50\pi$, $\omega_3=30.06\pi$, and $\omega_4=40.22\pi$ rad/s.
We sample with $T_s = 0.03$s, which is slightly faster than what Theorem~\ref{thm:uniformsampling} prescribes, and we set $t_{\max}=6.03$s.
After constructing the sampled data matrix $[V]$, we compute the SVD.
In order to estimate the modal frequencies we focus on the output matrix $[\hat{S}]$.
Referring to \eqref{eqn:svdofv}, we see that each row of $[\hat{S}]$ should approximately contain a complex exponential corresponding to one of the four modal frequencies.
One simple way to extract each frequency is to compute the Fourier transform for each row of $[\hat{S}]$ and to identify the frequency with maximum magnitude.
Figure~\ref{fig:usrs3} shows the magnitude plot of the FFT for each row of $[\hat{S}]$.
Extracting the peak from each row, we obtain frequency estimates of $\hat{\omega}_1=6.3018\pi$, $\hat{\omega}_2=20.5638\pi$, $\hat{\omega}_3=30.1824\pi$, $\hat{\omega}_4=40.4643\pi$ rad/s.
The accuracy of these estimates is naturally limited in that the frequency resolution will be inversely proportional to $t_{\max}$.
Finally, we note that when estimating modal frequencies, it is important that we avoid aliasing in $[S]$ in order to identify the correct modal frequencies.
This experiment represents a promising proof of concept that modal frequencies can be estimated from compressive measurements; we reserve more detailed analysis of this problem for a future paper.

\begin{figure}[t]
\centerline{
\begin{minipage}[t]{0.33\linewidth}
  \centering
  \centerline{\includegraphics[width=7cm]{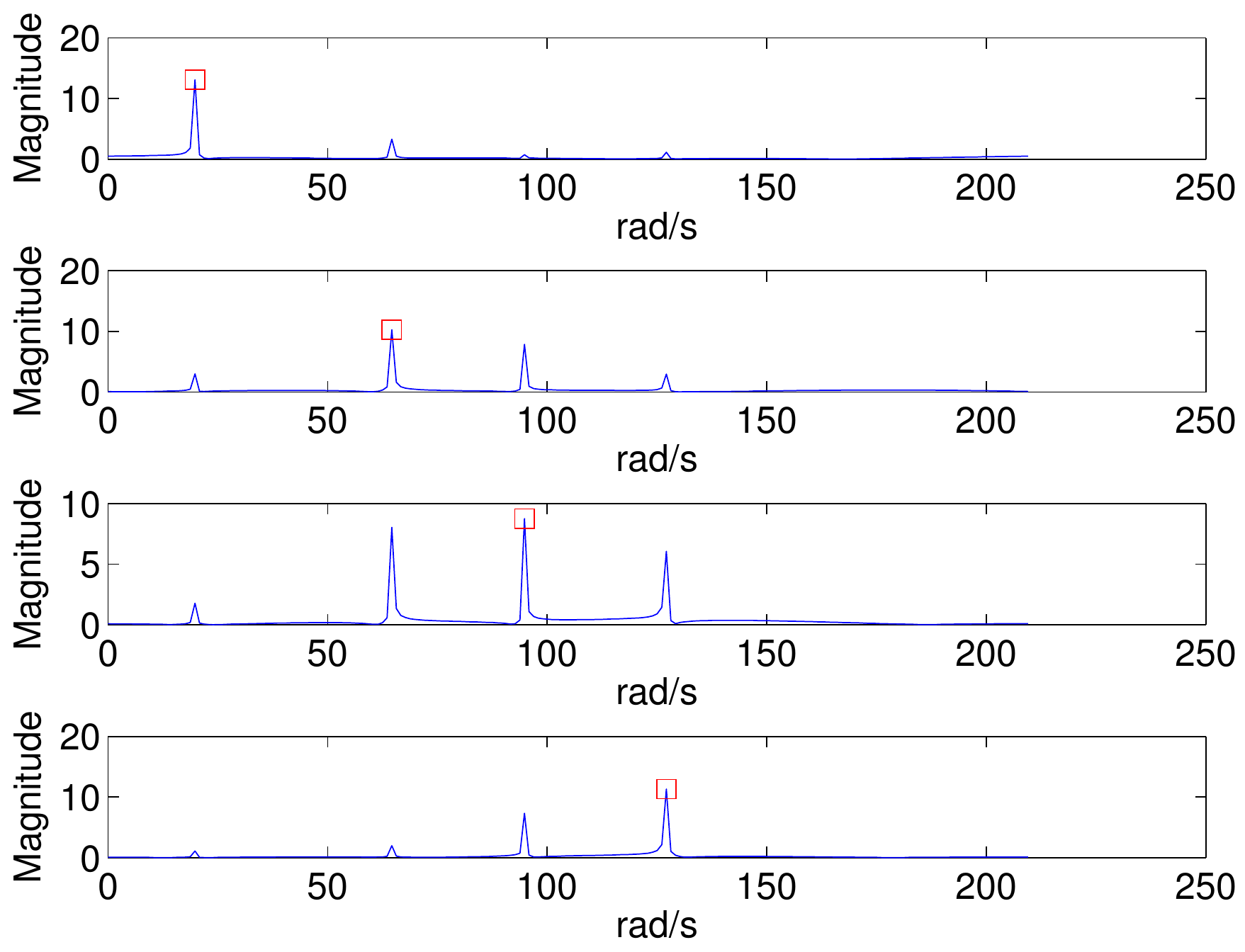}}
  \medskip
\end{minipage}
}
\caption{\small\sl {
FFT magnitude plots for the four rows of $[\hat{S}]$.
The top subplot corresponds to the first modal frequency, the second subplot corresponds to the second modal frequency, and so on.
The estimated modal frequencies were obtained by picking the largest peak in each plot (marked with red squares).
}%
}
\label{fig:usrs3}
\end{figure}

\subsection{Experiments with real data}

We conclude by presenting simulation results using vibration data collected from a bridge in Ypsilanti, MI.
On this bridge, there are $N=18$ wireless nodes, each of which is equipped with an accelerometer.
The relative layout of the sensors is shown in Figure~\ref{csshm:fig:bridge}.
Each sensor measures acceleration data and transmits it to the central node for analysis.
Because this data contains noise and is collected from a real bridge, which does have damping and is not necessarily in free vibration mode, we feel this represents an interesting test for our methods when the assumptions of our current theoretical results are violated.

\begin{figure}[t]
\begin{center}
\begin{tikzpicture}[yscale=1,xscale=1]
\draw [thick] (0,0)--(7,0);%
\draw [thick] (0,0)--(1,1.5);
\draw [thick] (1,1.5)--(8,1.5);%
\draw [thick] (7,0)--(8,1.5);
\draw [blue,fill] (1.6,1.35) circle (0.02);
\node [below] at (1.6,1.35) {\footnotesize  1};
\draw [blue,fill] (1.6+7/10*1,1.35) circle (0.02);
\node [below] at (1.6+7/10*1,1.35) {\footnotesize  2};
\draw [blue,fill] (1.6+7/10*2,1.35) circle (0.02);
\node [below] at (1.6+7/10*2,1.35) {\footnotesize  3};
\draw [blue,fill] (1.6+7/10*3,1.35) circle (0.02);
\node [below] at (1.6+7/10*3,1.35) {\footnotesize  4};
\draw [blue,fill] (1.6+7/10*4,1.35) circle (0.02);
\node [below] at (1.6+7/10*4,1.35) {\footnotesize  5};
\draw [blue,fill] (1.6+7/10*5,1.35) circle (0.02);
\node [below] at (1.6+7/10*5,1.35) {\footnotesize  6};
\draw [blue,fill] (1.6+7/10*6,1.35) circle (0.02);
\node [below] at (1.6+7/10*6,1.35) {\footnotesize  7};
\draw [blue,fill] (1.6+7/10*7,1.35) circle (0.02);
\node [below] at (1.6+7/10*7,1.35) {\footnotesize  8};
\draw [blue,fill] (1.6+7/10*8,1.35) circle (0.02);
\node [below] at (1.6+7/10*8,1.35) {\footnotesize  9};
\draw [blue,fill] (0.8,0.15) circle (0.02);
\node [above] at (0.8,0.15) {\footnotesize  10};
\draw [blue,fill] (0.8+7/10*1,0.15) circle (0.02);
\node [above] at (0.8+7/10*1,0.15) {\footnotesize  11};
\draw [blue,fill] (0.8+7/10*2,0.15) circle (0.02);
\node [above] at (0.8+7/10*2,0.15) {\footnotesize  12};
\draw [blue,fill] (0.8+7/10*3,0.15) circle (0.02);
\node [above] at (0.8+7/10*3,0.15) {\footnotesize  13};
\draw [blue,fill] (0.8+7/10*4,0.15) circle (0.02);
\node [above] at (0.8+7/10*4,0.15) {\footnotesize  14};
\draw [blue,fill] (0.8+7/10*5,0.15) circle (0.02);
\node [above] at (0.8+7/10*5,0.15) {\footnotesize  15};
\draw [blue,fill] (0.8+7/10*6,0.15) circle (0.02);
\node [above] at (0.8+7/10*6,0.15) {\footnotesize  16};
\draw [blue,fill] (0.8+7/10*7,0.15) circle (0.02);
\node [above] at (0.8+7/10*7,0.15) {\footnotesize  17};
\draw [blue,fill] (0.8+7/10*8,0.15) circle (0.02);
\node [above] at (0.8+7/10*8,0.15) {\footnotesize  18};
%
\end{tikzpicture}
\end{center}
\caption{\small\sl
%
%
The relative layout of 18 vibration sensors installed on a bridge in Ypsilanti, MI. From these sensors, we acquire vibration signals $\{u(t)\}(1), \dots, \{u(t)\}(18)$.
}
\label{csshm:fig:bridge}
\end{figure}
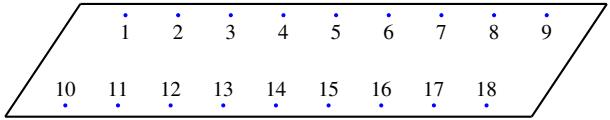

The data that is available to us from each sensor is a set of real-valued $M=3000$ samples collected uniformly in time at a rate faster than the Nyquist rate.
We stack this data into a real-valued data matrix we call $[U]$ (note that we do not assume samples of the analytic signals are available), and we test the effectiveness of multiplying $[U]$ by a random matrix $[\Phi]$ and then computing the SVD of the compressed matrix $[Y] = [U][\Phi]$.
For a point of comparison, we test a method that is similar to one presented in~\cite{CSDD2011,YuequanBao01052011}.
That method, which we refer to as ``CS+FDD,'' involves reconstructing each length-$M$ signal $\{u_l\}$ from the random Gaussian measurements $\{y_l\}=[\Phi]^* \{u_l\}$.
These signals are reconstructed one-by-one and then fed into the Frequency Domain Decomposition (FDD) method~\cite{Brincker_Zhang_Andersen_2000} for modal analysis.
The reconstruction of each $\{u_l\}$ is accomplished by solving
\begin{equation*}
\min_{\{\alpha_l\}}\|\{\alpha_l\}\|_1~\text{s.t.}~\{y_l\}=[\Phi]^*[W]\{\alpha_l\},
\end{equation*}
where $[W]$ represents a DWT matrix (we saw similar results with the DFT), and we let the reconstructed $\{u_l\}=[W]\{\alpha_l\}$.
%

For both the proposed SVD method and the CS+FDD method, we take $M'=50$ measurements of each $\{u_l\}$ using the same measurement matrix $[\Phi]$ for all $l$.
Since we do not know the true mode shapes of the structure, we use as a benchmark the three dominant mode shapes returned when FDD is applied to the original uncompressed data matrix $[U]$.
The results are presented in Figure~\ref{csshm:fig:CSFDD01}.
As we can see, the mode shapes estimated using CS+FDD (plotted in black) are not particularly close to the mode shapes returned when FDD is applied on the original data (plotted in blue).
This is apparently because it is difficult to accurately reconstruct any individual signal $\{u_l\}\in\mathbb{R}^{3000}$ from just $50$ random measurements.
In contrast, however, when we apply the SVD to the compressed data matrix $[Y]$, the estimated mode shapes (plotted in red) provide much better approximations to the true FDD mode shapes.

\begin{figure*}[t]
\centerline{
\begin{minipage}[t]{0.33\linewidth}
  \centering
  \centerline{\includegraphics[width=5cm]{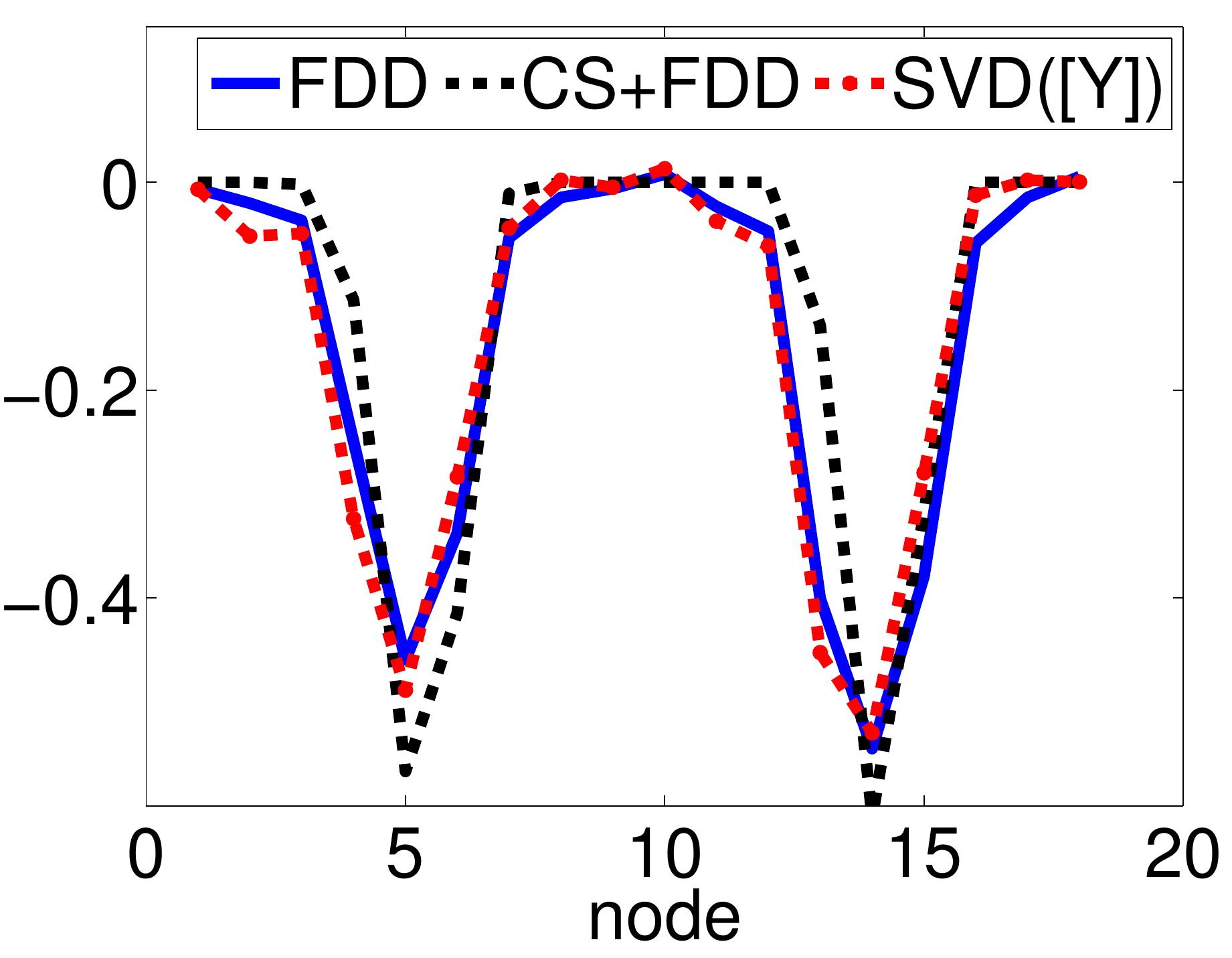}}
  \centerline{\small{(a)}}\medskip
\end{minipage}
\hfil
\begin{minipage}[t]{0.33\linewidth}
  \centering
  \centerline{\includegraphics[width=5cm]{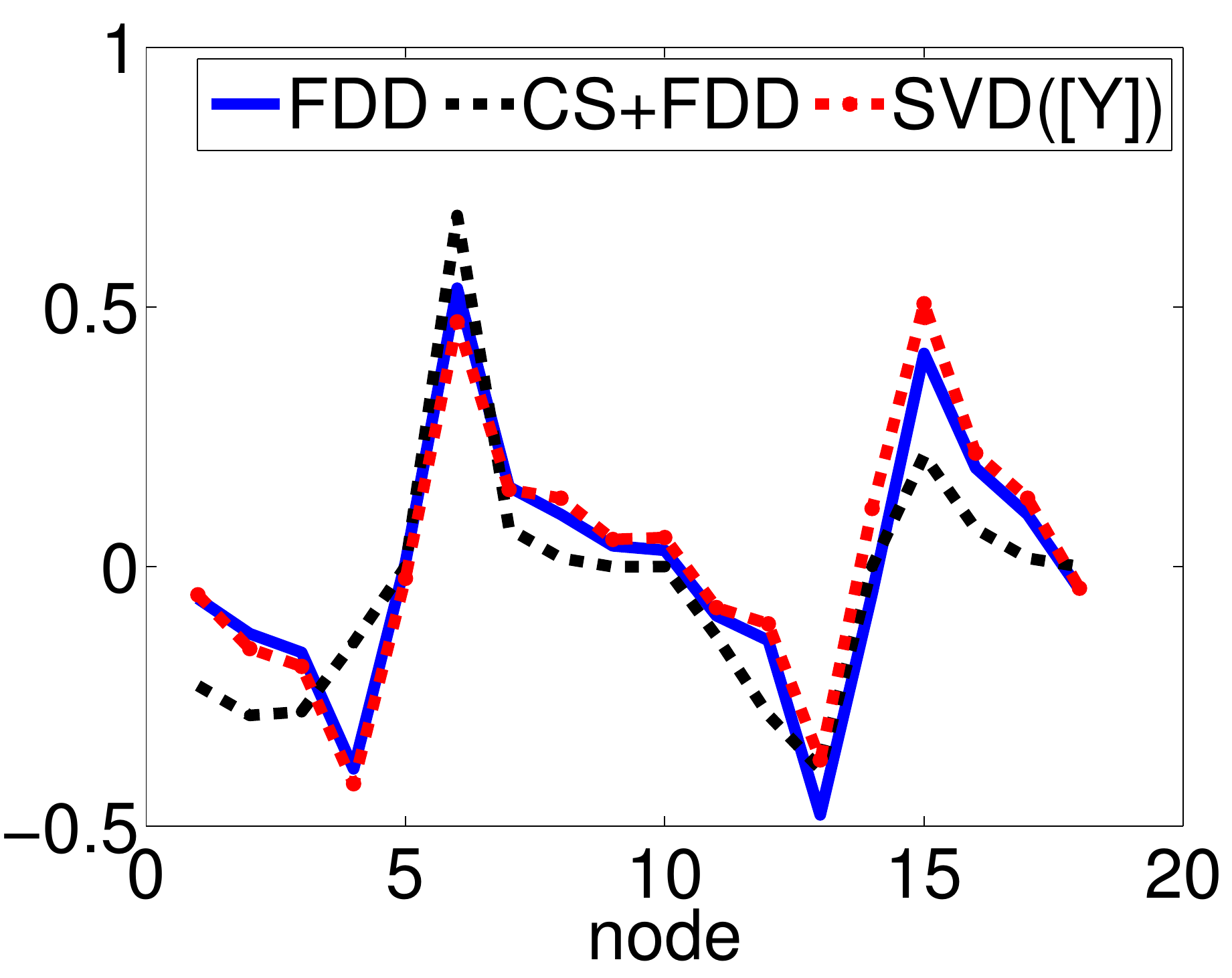}}
  \centerline{\small{(b)}}\medskip
\end{minipage}
\hfil
\begin{minipage}[t]{0.33\linewidth}
  \centering
  \centerline{\includegraphics[width=5cm]{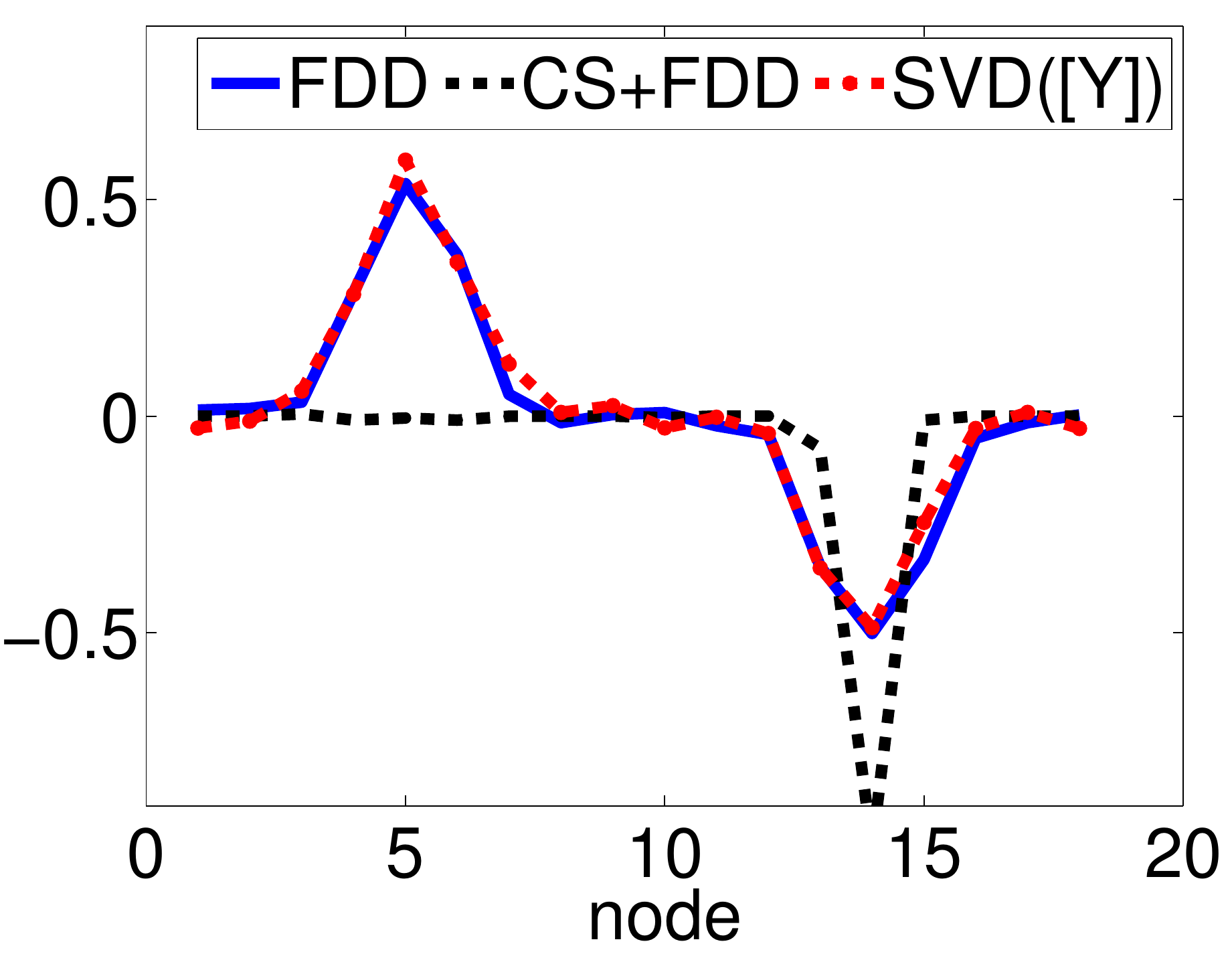}}
  \centerline{\small{(c)}}\medskip
\end{minipage}}
\vspace{-0.2cm}
\caption{\small\sl {Comparison of FDD on the original data matrix with CS+FDD and our proposed method \text{SVD}($[Y]$). Each mode shape returned by CS+FDD and \text{SVD}($[Y]$) is compared against the corresponding mode shape returned by FDD. The results were evaluated by computing the $\ell_2$ distance between the true and the estimated mode shapes. For each of the three dominant FDD mode shapes, the $\ell_2$ distance to the mode shape estimated from compressive measurements is as follows: (a) CS+FDD: $0.35$, \text{SVD}($[Y]$): $0.16$ (b) CS+FDD: $0.96$, \text{SVD}($[Y]$): $0.14$, and (c) CS+FDD: $0.50$, \text{SVD}($[Y]$): $0.19$.}%
}
\label{csshm:fig:CSFDD01}
\end{figure*}

These simulation results indeed support our theoretical results that the SVD of the data matrix $[Y]$ can return accurate estimates of the true mode shape vectors.
We emphasize again that the dataset in this simulation is real-valued, contains noise, and is collected from a real structure {\em with} damping; technically, none of this is covered by the assumptions of our current theoretical results.
The fact that our method was nevertheless able to successfully estimate the mode shape vectors is very encouraging and suggests that our theoretical findings may be extendable to more complicated scenarios.

\section*{Acknowledgment}

The authors would like to thank Sean O'Connor and Prof.~Jerome P. Lynch at the University of Michigan for helpful discussions on modal analysis and for providing us with real datasets to carry out the experiments presented in this paper.

\appendix

In this appendix we provide proofs of the main results.
To do so, we take a perturbation theoretic viewpoint.
We first describe how we can formulate our problem as a perturbation problem and then provide separate proofs for each theorem.

\subsection{Perturbation analysis}

We start with the equation $[V]=[\Psi][\Gamma][S]$, and we allow the sample times $t_1,\dots,t_M$ to be arbitrary.
To carry out perturbation analysis let us note that $[V][V]^{*}=[\Psi][\Gamma][S][S]^{*}[\Gamma]^{*}[\Psi]^{*}$, where $[S][S]^{*}$ is an $N \times N$ matrix with entries
\begin{equation*}
([S][S]^{*})_{l,n} = \left\{ \begin{array}{ll} 1, & l = n, \\ \frac{1}{M}\sum_{m=1}^{M} e^{i(\omega_l-\omega_n)t_m}, & l \neq n. \end{array} \right.
\end{equation*}
Thus, we can decompose this product as $[S][S]^{*}=[I]+[\Delta]$, where $[\Delta]$ contains the off-diagonal entries of $[S][S]^{*}$.
Then
$$
[V][V]^{*}=\underbrace{[\Psi][\Gamma][\Gamma]^{*}[\Psi]^{*}}_{[H]}+\underbrace{[\Psi][\Gamma][\Delta][\Gamma]^{*}[\Psi]^{*}}_{[\delta H]}.
$$
The above expression allows us to view $[V][V]^{*}$ as the summation of a matrix $[H]$ and a matrix $[\delta H]$.
We may view $[\delta H]$ as the perturbation matrix that is being added to $[H]$.
Noting that the eigenvectors of $[H]$ are given by $[\Psi]$, our goal is to show that the eigenvectors of $[V][V]^{*}$ (which equal the left singular vectors of $[V]$) are close to those of $[H]$.
To do this, we follow the approach in~\cite{sketchedgpw12} and employ~\cite[Theorem~1]{Mathias1998}, which provides a perturbation bound on the eigenvectors as a function of the quantity
\begin{align*}
\eta&:=\|[H]^{-\frac{1}{2}}[\delta H][H]^{-\frac{1}{2}}\|_2 \\
&=\| [\Psi]\left([{\Gamma}]  [{\Gamma}]^{*} \right)^{-\frac{1}{2}}[\Psi]^{*}
[\Psi][{\Gamma}] [\Delta] [{\Gamma}]^{*} [\Psi]^{*}
[\Psi] \cdot \\
& \hspace{0.22in} \left([{\Gamma}]  [{\Gamma}]^{*} \right)^{-\frac{1}{2}}[\Psi]^{*}\|_2\\
&=\|\left([{\Gamma}]  [{\Gamma}]^{*} \right)^{-\frac{1}{2}}
[{\Gamma}] [\Delta] [{\Gamma}]^{*}
\left([{\Gamma}]  [{\Gamma}]^{*} \right)^{-\frac{1}{2}}\|_2\\
&=\max_{\{x\}\neq 0}\frac{\{x\}^*\left([{\Gamma}]  [{\Gamma}]^{*} \right)^{-\frac{1}{2}}
[{\Gamma}] [\Delta]^{*} [\Delta] [{\Gamma}]^{*}
\left([{\Gamma}]  [{\Gamma}]^{*} \right)^{-\frac{1}{2}}\{x\}}{\{x\}^*\{x\}}\\
&=\max_{\{y\}\neq 0}\frac{\{y\}^* [\Delta]^{*} [\Delta]\{y\}}{\{y\}^*\{y\}}
\\
&=\| [\Delta] \|_2.
\end{align*}
In the second to last line we let $\{y\}=[{\Gamma}]^{*}\left([{\Gamma}]  [{\Gamma}]^{*} \right)^{-\frac{1}{2}}\{x\}$ and note that $\{y\}^{*}\{y\}=\{x\}^{*}\{x\}$.

Our Theorems~\ref{thm:uniformsampling} and~\ref{thm:randomsamplingFinal} follow by proving that $\| [\Delta] \|_2 \le \epsilon$.
More complete details on how we can apply~\cite[Theorem~1]{Mathias1998} are provided in~\cite{sketchedgpw12}.
We do make two notes here concerning the application of~\cite[Theorem~1]{Mathias1998}.
First, in order to apply this theorem, we require that $[H]$ be positive definite.
This leads to the requirement that $M \ge N$ in all of our results.
Second, in order to obtain a final bound that depends only on the eigenvalues of $[H]$ and not also on those of $[V][V]^{*}$, it is necessary to prove that the the eigenvalues of $[V][V]^{*}$ approximate those of $[H]$.
This fact also follows by proving that $\| [\Delta] \|_2 \le \epsilon$ and by applying~\cite[Lemma~2]{Barlow80computingaccurate}.

To compute an upper bound for $\|[\Delta]\|_2$, let us note that $\lambda_n([S][S]^{*})=\lambda_n([I]+[\Delta])=1+\lambda_n([\Delta])$, where we use $\lambda_n([A])$ to denote the $n$th largest eigenvalue of $[A]$.
If we can find upper and lower bounds on the eigenvalues of $[S][S]^{*}$ such that $\lambda_d\leq\lambda_n([S][S]^{*})\leq \lambda_u$ holds for all $n$,
then we can see that $\lambda_d-1\leq \lambda_n([\Delta]) \leq \lambda_u-1$, which in turn implies that $\|[\Delta]\|_2=\max_n |\lambda_n([\Delta])|\leq \max\{|\lambda_u-1|,|\lambda_d-1|\}$.
Therefore, our strategy is to bound $\lambda_n([S][S]^{*})$ from below and above in order to bound $\| [\Delta] \|_2$.
In the following sections we establish this result for both the random and uniform sampling cases.

\subsection{Proof of Theorem~\ref{thm:randomsamplingFinal} (random sampling)}

Let us first consider how to establish a bound on $\|[\Delta]\|_2$ if we were to sample $t_1,\dots,t_{M}$ uniformly at random in the time interval $[0,t_{\max}]$.
We can establish the following bound on the eigenvalues of $[S][S]^{*}$.

\begin{theorem}
\label{thm:randomsampling}
Given that we sample with $t_{\max}$ satisfying~\eqref{eq:tmaxrandom} and $M$ satisfying \eqref{eq:Mrandom}, then with probability at least $1- \tau$ we will have $1-\epsilon<\lambda_n\left([S][S]^{*}\right)< 1+\epsilon$ for all $n$.
\end{theorem}

\begin{proof}
To bound $\|[S][S]^{*}\|_2$ we use a slightly modified version of a theorem that appeared in~\cite{JTtailbounds}.
\begin{theorem}\label{thm:JT}[\cite{JTtailbounds}, Theorem 5.1]
Consider a sequence $\{[Z_m]:m=1,\dots,M\}$ of independent, $d$-dimensional, random, self-adjoint matrices that satisfy $[Z_m]\succeq 0$ and $\lambda_{\max}([Z_m])\leq c$ almost surely.
Then for any $\tilde{\mu}_{\min}$ and $\tilde{\mu}_{\max}$ such that
\begin{align*}
\tilde{\mu}_{\min}&\leq\lambda_{\min}\left(\frac{1}{M}\sum_{m=1}^{M}\E [Z_m]  \right)~\text{and} \\
\tilde{\mu}_{\max}&\geq\lambda_{\max}\left(\frac{1}{M}\sum_{m=1}^{M}\E [Z_m]  \right),
\end{align*}
we have
\begin{equation}
P\left\{ \lambda_{\min}\left( \frac{1}{M}\sum_{m=1}^{M}[Z_m] \right) \leq \alpha    \right\}\leq d e^{-MD(\alpha/c ||\tilde{\mu}_{\min}/c)}
\label{eqn:jtlambdamin}
\end{equation}
for $0\leq\alpha\leq\tilde{\mu}_{\min}$ and
\begin{equation}
P\left\{ \lambda_{\max}\left( \frac{1}{M}\sum_{m=1}^{M}[Z_m] \right) \geq \alpha    \right\}\leq d e^{-MD(\alpha/c ||\tilde{\mu}_{\max}/c)}
\label{eqn:jtlambdamax}
\end{equation}
for $\tilde{\mu}_{\max}\leq\alpha\leq c$. 
\end{theorem}
In order to apply Theorem~\ref{thm:JT}, let us write
$$
[S][S]^{*} =\frac{1}{M}\sum_{m=1}^{M}
\left\{
\begin{array}{c}
e^{i\omega_1 t_m}\\
\vdots\\
e^{i\omega_N t_m}
\end{array}
\right\}
\left\{
e^{-i\omega_1 t_m},\dots,e^{-i\omega_N t_m}
\right\}.
$$
We define the vector $\{S_m\}=\{e^{-i\omega_1 t_m},\dots,e^{-i\omega_N t_m}\}^{*}$, where $\|\{S_m\}\|_2^2=N$.
Let us set the matrix $[Z_m]$ that appears in the above theorem as $[Z_m]=\{S_m\}\{S_m\}^{*}$.
As a result, $[Z_m]$ will be i.i.d. positive semi-definite matrices, i.e., $[Z_m]\succeq 0$, of rank 1 with $\lambda_{\max}([Z_m])=\|\{S_m\}\|_2^2= N$.
We wish to compute
$$
\lambda_{\min}\left(\frac{1}{M}\sum_{m=1}^M \E \{S_m\} \{S_m\}^{*} \right)=\lambda_{\min}\left( \E \{S_m\} \{S_m\}^{*} \right),
$$
$$
\lambda_{\max}\left(\frac{1}{M}\sum_{m=1}^M \E \{S_m\} \{S_m\}^{*} \right)=\lambda_{\max}\left( \E \{S_m\} \{S_m\}^{*} \right),
$$
or the appropriate lower and upper bound on the above quantities.
Note that
\begin{align*}
& (\E \{S_m\} \{S_m\}^{*})_{l,n} =
\left\{ \begin{array}{ll} 1, & l = n, \\ \E e^{i(\omega_l-\omega_n)t_m}, & l \neq n \end{array} \right. \\
& ~~~ =
\left\{
\begin{array}{ll}
1, & l=n,\\
e^{i(\omega_l-\omega_n)\frac{t_{\max}}{2}}\sinc((\omega_l-\omega_n)\frac{t_{\max}}{2}),& l\neq n .
\end{array}
\right.
\end{align*}
The eigenvalues can also be written as $\lambda_{n}\left( \E \{S_m\} \{S_m\}^{*} \right)=\lambda_{n}\left( I+[\Delta_S] \right)=1+\lambda_{n}([\Delta_S])$, where $[\Delta_S]$ is the off-diagonal matrix of $\E \{S_m\} \{S_m\}^{*}$.
According to Gershgorin's circle theorem~\cite{ger31} we know that every eigenvalue of $[\Delta_S]$ must lie within at least one Gershgorin disk.
As $[\Delta_S]$ has zero diagonal entries, every Gershgorin disk must be centered at zero.
Thus, the radius of the largest disk will provide a bound on all eigenvalues of $[\Delta_S]$.
It follows that every eigenvalue of $[\Delta_S]$ will obey the following bound:
\begin{align}
|\lambda ([\Delta_S])| &\leq \max_{l} \sum_{n=1,n\neq l}^N \left| \sinc((\omega_l-\omega_n)\frac{t_{\max}}{2})\right| \nonumber \\
&\leq \max_{l} \sum_{n=1,n\neq l}^N \frac{2}{|\omega_l-\omega_n|t_{\max}} \nonumber \\
&\leq \sum_{n=1,n\neq l'}^N \frac{2}{|\omega_{l'}-\omega_n|t_{\max}}\nonumber \\
&\leq\frac{4}{\delta_{\min} t_{\max}} \sum_{n=1}^{\lfloor N/2\rfloor} \frac{1}{n}, \label{eq:lamDel}
\end{align}
where we have denoted the index of the middle row of $[\Delta_S]$ as $l'$ (when $N$ is even we can take either of $N/2$ or $N/2+1$ as the middle row), and we have used the fact that $\omega_n - \omega_{n+1} \ge \delta_{\min}$ for $n = 1,\dots,N-1$.
The summation term in the above bound is also known as the Harmonic number. We can simplify the above expression by using the following bound on the Harmonic number.
\begin{theorem}[\cite{HarmonicBound}, Theorem 1]
For any natural number $N\geq 1$, the following inequality is valid:
\begin{equation*}
\frac{1}{2N+\frac{1}{1-\gamma}-2}\leq \sum_{n=1}^{N} \frac{1}{n} -\log(N)-\gamma <\frac{1}{2N+\frac{1}{3}}.
\end{equation*}
The constant $\gamma=0.57721\cdots$ is known as Euler's constant. The constants $\frac{1}{1-\gamma}-2=0.3652\cdots$ and $\frac{1}{3}$ are the best possible, and equality holds only for $N=1$.
\end{theorem}
Applying this theorem to \eqref{eq:lamDel}, we have
\begin{align*}
|\lambda([\Delta_S])|&< \frac{4(\log\lfloor N/2 \rfloor+\gamma +\frac{1}{2\lfloor N/2\rfloor+\frac{1}{3}})}{\delta_{\min} t_{\max}} \\
&< \frac{4(\log\lfloor N/2 \rfloor+\gamma +3/7) }{\delta_{\min} t_{\max}} \\
&< \frac{4(\log\lfloor N/2 \rfloor+1.01) }{\delta_{\min} t_{\max}} .
\end{align*}
Collecting everything together, we will have for all $n$,
\begin{align*}
1- \frac{4(\log\lfloor N/2 \rfloor+1.01) }{\delta_{\min} t_{\max}} &<\lambda_n\left( \E \{S_m\} \{S_m\}^{*} \right) \\ &=1+\lambda_n([\Delta_S]) \\ &< 1+ \frac{4(\log\lfloor N/2 \rfloor+1.01) }{\delta_{\min} t_{\max}}. 
\end{align*}
Supposing that~\eqref{eq:tmaxrandom} is satisfied, we have that
$$
\tilde{\mu}_{\min}:=1-\frac{\epsilon}{10}<\lambda_n(\E \{S_m\}\{S_m\}^{*})< 1+\frac{\epsilon}{10}=:\tilde{\mu}_{\max}.
$$
Note that $\tilde{\mu}_{\max}\leq 1+\epsilon$ and $\tilde{\mu}_{\min}\geq 1-\epsilon$. Then, according to the above theorem, inequality~\eqref{eqn:jtlambdamax} will hold for any $1+\frac{\epsilon}{10}\leq\alpha\leq N$, which will always include $\alpha=1+\epsilon$.
Similarly, inequality~\eqref{eqn:jtlambdamin} will hold for any $0\leq\alpha\leq 1-\frac{\epsilon}{10}$, which will always include $\alpha=1-\epsilon$.
Substituting the appropriate values of $\alpha=1\pm\epsilon$, $\tilde{\mu}_{\max}$, and $\tilde{\mu}_{\min}$ into Theorem~\ref{thm:JT}, with probability at least $1-Ne^{-MD_1}-Ne^{-MD_2}$ we will have
\begin{equation*}
1-\epsilon<\lambda_{\min}([S][S]^{*})\leq\lambda_n([S][S]^{*}) \leq\lambda_{\max}([S][S]^{*})< 1+\epsilon.
\end{equation*}
By choosing $M$ to satisfy~\eqref{eq:Mrandom}, we can ensure both that $Ne^{-MD_1} < \tau/2$ and that $Ne^{-MD_2} < \tau/2$, and therefore the overall failure probability will be less than $\tau$.
\end{proof}

\subsection{Proof of Theorem~\ref{thm:uniformsampling} (uniform sampling)}

For the uniform sampling scenario, we can establish the following theorem on the eigenvalues of $[S][S]^{*}$.
\begin{theorem}
Suppose we sample with a total time span $t_{\max}$ satisfying~\eqref{eq:tmaxuniform} with sampling interval $T_s=\frac{\pi}{\delta_{\max}}$ and ensure that $M \ge N$. Or, equivalently, suppose we take $M$ total samples with $M$ satisfying~\eqref{eq:Muniform} and with sampling interval $T_s=\frac{\pi}{\delta_{\max}}$.
Then we establish the following bound on the eigenvalues of $[S][S]^{*}$: $1-\epsilon\leq\lambda_n([S][S]^{*})\leq1+\epsilon$.
\end{theorem}
\begin{proof}
The off-diagonal matrix of $[S][S]^{*}$, denoted as $[\Delta]$, has the following entries: $[\Delta]_{l,n} = 0$ when $l = n$, and
\begin{align*}
[\Delta]_{l,n}
&=\frac{1}{M}\sum_{m=0}^{M-1}e^{i(\omega_p-\omega_q)mT_s}
=\frac{1}{M}\frac{1-e^{i(\omega_l-\omega_n)MT_s}}{1-e^{i(\omega_l-\omega_n)T_s}}\\
&=e^{i(\omega_l-\omega_n)T_s(M-1)/2}\frac{\sin((\omega_l-\omega_n)\frac{MT_s}{2})}{M\sin((\omega_l-\omega_n)\frac{T_s}{2})}
\\ &=e^{i(\omega_l-\omega_n)T_s(M-1)/2}\frac{\sin(|\omega_l-\omega_n|\frac{MT_s}{2})}{M\sin(|\omega_l-\omega_n|\frac{T_s}{2})}
\end{align*}
when $l \neq n$.
The fraction of sinusoids in the above equation is known as the periodic sinc function or the Dirichlet function and is defined as $\psinc(x)=\frac{\sin(M\frac{x}{2})}{M\sin(\frac{x}{2})}$.
More specifically,
\begin{equation*}
\psinc(x)=
\left\{
\begin{array}{ll}
\frac{\sin(M\frac{x}{2})}{M\sin(\frac{x}{2})}, & x\neq 2\pi k, ~ k=0,\pm1,\pm2,\dots \\
(-1)^{k(M-1)}, & x= 2\pi k, ~ k=0,\pm1,\pm2,\dots.
\end{array}
\right.
\end{equation*}
As its name implies, the $\psinc$ function is a periodic function where the period is equal to $2\pi$ when $M$ is odd, and $4\pi$ when $M$ is even.
Every time $x$ is equal to an integer multiple of $2\pi$, the $\psinc$ function will reach its maximum value.

Again, we bound the eigenvalues of $[\Delta]$ using Gershgorin's disk theorem.
Since every Gershgorin disk will again be centered at zero, every eigenvalue of $[\Delta]$ must obey the following bound:
\begin{equation}
|\lambda([\Delta])| \leq 
\max_{l} \sum_{n=1, n \neq l}^N |\psinc(|\omega_l-\omega_n| T_s)|.
\label{eqn:uniformDelta111}
\end{equation}
We can guarantee the evaluation of the $\psinc$ function to be small by restricting ourselves to only certain values of $T_s$ and $M$.
We compute an upper bound on the $\psinc$ function by noting that the denominator can be lower bounded by a linear function for a certain range of $T_s$.
In particular, for $T_s\leq\frac{\pi}{|\omega_l-\omega_n|}$ we have that $\sin\left(|\omega_l-\omega_n|\frac{T_s}{2} \right)\geq \frac{|\omega_l-\omega_n|T_s}{\pi}$.
Applying this to \eqref{eqn:uniformDelta111}, we have
\begin{align*}
|\lambda([\Delta])| &\leq \max_{l} \sum_{n=1, n \neq l}^N  \left|\frac{\pi\sin(|\omega_l-\omega_n|\frac{MT_s}{2})}{|\omega_l-\omega_n|MT_s}\right|
\\ &= \max_{l} \sum_{n=1, n \neq l}^N \frac{\pi}{2}\left|\sinc(|\omega_l-\omega_n|\frac{MT_s}{2})\right| \\
&\le \max_{l} \sum_{n=1, n \neq l}^N \frac{\pi}{|\omega_l-\omega_n|MT_s}
\end{align*}
when $T_s\leq \frac{\pi}{|\omega_l-\omega_n|}$ for all $l$ and $n$. In the last line we have used the fact that $|\sinc(x)|\leq1/|x|$. To ensure that $T_s\leq \frac{\pi}{|\omega_l-\omega_n|}$ for all $l$ and $n$, we choose the sampling interval such that $T_s \leq \frac{\pi}{\delta_{\max}}$.
Following the same arguments we used in the proof of Theorem~\ref{thm:randomsamplingFinal}, this will give us
\begin{align*}
|\lambda([\Delta])|
&\leq\frac{\pi}{\delta_{\min}MT_s}
\sum_{n=1}^{\lfloor N/2\rfloor} \frac{2}{n}
<
\frac{2\pi(\log\lfloor N/2 \rfloor+1.01)}{\delta_{\min}MT_s}
 \\ &<
\frac{2\pi(\log\lfloor N/2 \rfloor+1.01)}{\delta_{\min}(M-1)T_s}.
\end{align*}
Now, note that $(M-1)T_s=t_{\max}$ and if we choose $t_{\max}$ so that
$$
t_{\max}\geq
\frac{2\pi(\log\lfloor N/2 \rfloor+1.01)}{\delta_{\min}\epsilon},
$$
we will have that
$
|\lambda([\Delta])|<\epsilon.
$
In summary, when we sample within a sampling interval satisfying $T_s\leq  \frac{\pi}{\delta_{\max}}$ and a time span satisfying
$t_{\max}\geq
\frac{2\pi(\log\lfloor N/2 \rfloor+1.01)}{\delta_{\min}\epsilon}$, we will have
$$
1-\epsilon<\lambda_n([S][S]^{*})=1+\lambda_n([\Delta])<1+\epsilon.
$$
Or, in other words, if we set $T_s=\frac{\pi}{\delta_{\max}}$, and remembering that $t_{\max}=(M-1)T_s$, this means that
we need to sample at least
$$
M\geq\frac{2(\log\lfloor N/2 \rfloor+1.01)}{\epsilon}\frac{\delta_{\max}}{\delta_{\min}}+1,
$$
to achieve the above eigenvalue guarantee. Note that this is the smallest number of measurements we need since we set $T_s$ as large as possible.
If we were to reduce the sampling interval we would need to take more measurements to cover the same time span $t_{\max}$.
\end{proof}

\subsection{Proof of Theorem~\ref{thm:uniformandrandom} (uniform sampling with matrix multiplication)}

This result can be shown by simply using the triangle inequality. Let us write
$$
\|\{\psi_n\}-\{\tilde{\psi}_n\}\|_2\leq \|\{\psi_n\}-\{\hat{\psi}_n\}\|_2+\|\{\hat{\psi}_n\}-\{\tilde{\psi}_n\}\|_2,
$$
where $\{\hat{\psi}_1\},\dots,\{\hat{\psi}_N\}$ denote the left singular vectors of $[Y]$.
Note that the first term on the right hand side of the above inequality represents the error in the mode shape vectors due to the uniform sampling scheme as presented in Theorem~\ref{thm:uniformsampling}. The second term represents the difference between the left singular vectors of $[V]$ and those of $[Y]$. To quantify the amount of this error we make use of~\cite[Theorem~1]{sketchedgpw12}. Substituting the upper bound for each term completes the proof.

\bibliographystyle{IEEEtran}
\bibliography{refs-intro}

\end{document}